\documentclass[11pt]{article}
\usepackage[utf8]{inputenc}
\usepackage[a4paper, total={6.5in, 9in}]{geometry}
\usepackage{amsfonts, amsmath, amsthm, amssymb, mathtools}
\usepackage{graphicx, xcolor, natbib}
\usepackage{algorithm, algorithmicx, algpseudocode}
\usepackage{hyperref}
\usepackage{booktabs}
\usepackage[color]{changebar}

\usepackage{caption}
\newcommand{\indep}{\perp\!\!\!\!\perp}
\newcommand{\iid}{\stackrel{\text{iid}}{\sim}}
\newcommand{\ind}{\stackrel{\text{ind}}{\sim}}

\newcommand{\given}{\,\mid\,}
\newcommand{\agiven}{\;\middle|\;}  

\newcommand{\thetatrain}{\hat{\theta}^{(1)}}           

\newcommand{\ytrain}{y^{(1)}}                           
\newcommand{\ytest}{y^{(2)}}                            

\newcommand{\msehat}{\widehat{\mathrm{MSE}}_{\epsilon}}
\newcommand{\msebar}{\overline{\widehat{\mathrm{MSE}}}}
\newcommand{\msethin}{\mathrm{MSE}_{\epsilon}}
\newcommand{\msefull}{\mathrm{MSE}_{\mathrm{full}}}

\newcommand{\E}[2][]{\mathbb{E}_{#1}\!\left[\,#2\,\right]}

\newcommand{\Var}[2][]{\operatorname{Var}_{#1}\!\left[\,#2\,\right]}
\newcommand{\Varthin}[1]{\Var[\ytrain, \ytest]{#1}}

\newcommand{\Cov}[3][]{\operatorname{Cov}_{#1}\!\left[\,#2,\,#3\,\right]}
\newcommand{\Corr}[3][]{\operatorname{Corr}_{#1}\!\left[\,#2,\,#3\,\right]}

\newcommand{\N}[2]{N\!\left(#1, #2\right)}

\theoremstyle{plain}
\newtheorem{theorem}{Theorem}[section]   
\newtheorem{lemma}[theorem]{Lemma}       
\newtheorem{proposition}[theorem]{Proposition}
\newtheorem{corollary}[theorem]{Corollary}

\theoremstyle{definition}

\newtheorem{assumption}[theorem]{Assumption}

\theoremstyle{remark}
\newtheorem{remark}[theorem]{Remark}

\usepackage{xcolor}

\title{On Data Thinning for Model Validation in Small Area Estimation}
\usepackage{orcidlink}
\usepackage{authblk}

\author{Sho Kawano \orcidlink{0009-0008-1841-8305}\thanks{Corresponding author: \href{mailto:shkawano@ucsc.edu}{shkawano@ucsc.edu}}}
\author{Paul A. Parker \orcidlink{0000-0002-3206-6122}}
\author{Zehang Richard Li \orcidlink{0000-0001-9551-9638}}
\affil{Department of Statistics, University of California, Santa Cruz, CA 95064}
\date{\today}

\begin{document}

\maketitle

\begin{abstract}
Small area estimation produces estimates of population parameters for geographic and demographic subgroups with limited sample sizes. Such estimates are critical for policy decisions, yet principled validation of these models remains a challenge. Unlike conventional predictive settings, validation data are rarely available. Data thinning splits a single observation into independent training and test components. It enables out-of-sample validation using only the area-level summary statistics routinely available, requiring only their Gaussianity and known sampling variances. However, the properties of thinning-based model comparison have not been formally studied. In this paper, we develop these properties. We construct an unbiased estimator of thinned-data mean squared error and show that it differs systematically from its full-data counterpart; for the standard Fay--Herriot model, the gap admits a closed-form expression that depends on the candidate model's shrinkage behavior. We further show that the estimator variance increases sharply as the training fraction approaches one, producing a bias-variance tradeoff with no universally optimal thinning parameter. Practical recommendations balancing these forces are informed by theory and verified empirically. Design-based simulations using American Community Survey microdata show that the recommended data thinning approach is competitive with information-criterion and simulation-based methods, and substantially more stable across heterogeneous sampling designs.
\end{abstract}

\section{Introduction}

Small area estimation (SAE) provides critical information for policy makers and analysts throughout the world. From poverty mapping to disease surveillance to electoral analysis, practitioners rely on model-based estimators to produce reliable estimates for geographic or demographic subgroups where direct survey estimates are too imprecise. In the United States, the Census Bureau's Small Area Income and Poverty Estimates program is the primary source of annual income and poverty statistics for all states and counties; these estimates are used for administering federal programs and allocating funds that amounted to more than \$14 billion in 2013 \citep{pratesi_overview_2016}. Globally, the United Nations General Assembly’s 2030 Agenda established a set of Sustainable Development Goals for global development, requiring accurate tracking of demographic and health indicators in fine geographic resolutions \citep{general_assemby_of_the_united_nations_resolution_2015}.

Despite the importance of these applications, standard practices for validating and comparing SAE models are far from well-established. Model choice can substantially affect the estimates that inform consequential decisions, making validation an important open problem. In the SAE setting, external validation surveys or censuses often do not exist, and access to individual-level microdata is frequently restricted, with national statistical agencies releasing only area-level summaries rather than unit-level records. Even when microdata are accessible, the complexity of unit-level models can make area-level approaches preferable in practice. Analysts in this \textit{area-level modeling paradigm} must therefore validate models using a single set of survey summaries, with no independent replication and no unit-level observations to fall back on.

\subsection{Existing Approaches and Their Pitfalls}

The most credible validation approaches for SAE models require fairly special circumstances and data that eludes most data analysts. Census-based validation compares model estimates against census values treated as truth, providing an unambiguous benchmark. But census values are available for only a narrow set of variables, subject to temporal lag, and often entirely absent in low- and middle-income countries \citep{dong_toward_2025}. Design-based simulation studies generate repeated samples from survey microdata, fit models to each, and evaluate accuracy against known population quantities. This approach is the methodological gold standard when microdata are available (see \citet{molina_small_2010, datta_small_2015}), although it can incur heavy computational costs. When neither census data nor microdata are accessible, practitioners resort to alternative procedures, often not designed for comparing SAE models. We identify four recurring pitfalls that arise in these procedures.

\begin{itemize}
    \item \textbf{Paradigm dependence (A):} The procedure is tied to either the Bayesian or frequentist framework, complicating comparisons of estimators across paradigms.

    \item \textbf{In-sample evaluation (B):} The procedure uses the same data for fitting and assessment and relies on a penalty term to approximate out-of-sample performance. The approximation relies on asymptotic arguments that may not hold across all survey settings.

    \item \textbf{Holding out areas (C):} The procedure splits the data into training and testing set by holding out a subset of areas. Predicting for held-out areas effectively evaluates the ability to extrapolate, which is a different inferential goal than improving estimates in sampled areas.

    \item \textbf{Additional assumptions (D):} The procedure relies on strong extra assumptions on the true data-generating process that cannot be empirically verified, with no guarantee of validity when these assumptions are violated. 
\end{itemize}

Existing validation approaches in SAE often reflect one or more of these limitations. For example, information criteria are commonly used for model comparison in SAE. As in-sample measures (B), they add a complexity penalty to a goodness-of-fit term, but penalty estimation relies on asymptotic approximations in the number of areas, which can be moderate in SAE applications. They are also paradigm-dependent (A), e.g., AIC \citep{akaike_new_1974} arises from the frequentist paradigm, while measures like DIC \citep{spiegelhalter_bayesian_2002} and WAIC \citep{watanabe_asymptotic_2010} are defined for Bayesian models through posterior inference.

Area-level cross-validation offers a more direct out-of-sample assessment and has been used to assess models \citep{michal_model-based_2024}.  In area-level models, such evaluation is equivalent to leave-one-out cross-validation and the conditional predictive ordinate (CPO)
\citep{stern_posterior_2000, marshall_approximate_2003}, with efficient approximations available for Bayesian models \citep{vehtari_practical_2017}.
Unlike standard prediction problems, however, removing an area's direct estimate eliminates the only area-specific information available, effectively evaluating extrapolation rather than improving estimates in sampled areas.

Simulation-based assessment is also frequently used in the literature \citep{bradley_multivariate_2015, janicki_bayesian_2022}. A common practice is to generate synthetic direct estimates by adding noise based on survey variances to observed direct estimates, and then validate model fits against the original direct estimates. We term this approach Empirical Simulation (ESIM).\footnote{Appendix~\ref{app:esim} covers the connections between ESIM and data fission, a related technique to data thinning.} This treats the observed direct estimates as truth (D), an assumption difficult to justify where model-based estimation is most needed. Model-based simulation studies similarly assume a known data-generating process (D), making them valuable for controlled theoretical study but less directly informative for real-world performance. In both cases, it shifts the question from assessing models given the observed data to assessing models given a somewhat arbitrarily assumed underlying truth. Simulation studies also often incur a high computational burden.

The fence method \citep{jiang_fence_2008} takes a different approach: compute a lack-of-fit measure for each candidate, set a fence based on the minimum plus a margin, and select the simplest model within the fence. As an in-sample method (B), it requires calibrating a tuning constant via bootstrap, introducing computational and calibration challenges analogous to fold selection. Variants have been proposed to reduce computational burden for restricted maximum-likelihood methods \citep{nguyen_restricted_2012}.

We note two categories of methods outside our scope. First, unit-level validation approaches, including survey-weighted cross-validation \citep{wieczorek_k-fold_2022} and leave-one-out cross-validation at the unit-level \citep{kuh_using_2024}, require microdata and thus fall outside the area-level paradigm we address. Second, diagnostic tools such as influence measures \citep{marcis_three-fold_2023} and goodness-of-smoothing statistics \citep{duncan_comparing_2020} help identify potential problems with a fitted model but do not provide formal selection or validation criteria.

\subsection{Data Thinning for Area-level SAE Models}

In this paper, we develop a novel model comparison approach for area-level models in SAE, addressing each limitation identified above. Our approach is based on data thinning \citep{neufeld_data_2024, dharamshi_generalized_2025}; see also \citet{rasines_splitting_2023, leiner_data_2025, oliveira_unbiased_2024} for related work. Data thinning splits a single observation into two independent training
and test components that add up to the original.
For Gaussian data, data thinning relies on two assumptions: the Gaussian distributional assumption and known variance parameters.
Most area-level SAE models adopt a Gaussian likelihood and treat the sampling variances associated with direct survey estimates as known, making them well suited to data thinning.
The foundational area-level SAE model, the Fay--Herriot model \citep{fay_estimates_1979}, satisfies these conditions directly, together with the majority of its extensions including spatial random effects \citep{zhou_hierarchical_2008}, shrinkage priors on random effects \citep{datta_small_2015}, combined spatial and shrinkage priors  \citep{kawano_spatially_2025},  nonlinear mean structures \citep{parker_nonlinear_2024}, and spatially varying regression coefficients \citep{janicki_bayesian_2022}.

Data thinning addresses each of the limitations (A--D) discussed above. It is estimator-agnostic (A), applying equally to Bayesian and frequentist approaches by comparing predictions to genuinely independent test data. It avoids in-sample bias (B) because training and test sets are marginally independent, requiring no penalty terms to approximate out-of-sample performance. It improves upon area-level cross-validation (C) by providing continuous control over training fractions while keeping all areas in the training and test components. Finally, it requires only standard modeling assumptions (D): Gaussian direct estimates with known sampling variances. Our empirical analysis demonstrates that data thinning
yields reliable model comparisons that are competitive with DIC, WAIC, and ESIM while providing much more stable performance across different sampling designs.

We also provide, to our knowledge, the first theoretical characterization of the fundamental properties of data thinning when validating SAE models. First, we identify a systematic discrepancy between thinned-data and full-data performance metrics. We show that this gap depends on model complexity through shrinkage parameters and creates bias toward simpler models when information allocated to the training component is low. Second, we show that the variance of the estimated model performance metric increases when more information is allocated to the training component. These competing forces reveal a fundamental tension for validation using data thinning, implying no universally optimal thinning parameter exists across candidate models.

More broadly, the assumptions enabling data thinning for SAE models also appear in other latent Gaussian settings such as meta-analysis, measurement error models, and spatial statistics with instrument-level precision, so the implications from our analysis are not limited to SAE. Our findings characterize data thinning properties that arise whenever model complexity affects shrinkage behavior.

The remainder of the paper proceeds as follows. Section~\ref{sec:background} reviews the Fay--Herriot model, introduces Gaussian data thinning, and presents our motivating example of spatial basis function selection. Section~\ref{sec:mse_validation} develops theoretical results for MSE-based validation, establishing unbiased estimation, analyzing the thinning gap, and characterizing the variance-gap trade-off. Section~\ref{sec:repeated} compares repeated and multi-fold thinning strategies. Section~\ref{sec:likelihood} extends the framework to likelihood-based validation, showing connections to weighted MSE. Section~\ref{sec:empirical} presents empirical results comparing data thinning against existing methods across multiple survey designs. Section~\ref{sec:discussion} concludes with discussion of limitations and extensions.

\section{Background} \label{sec:background}

\subsection{Small Area Estimation and the Fay--Herriot Model}

 Let a finite population be partitioned into $m$ small areas. In each area $i=1,\dots,m$, a survey sample of size $n_i$ is drawn from a population of size $N_i$, with the goal of estimating the finite population means $\theta_1, \ldots, \theta_m$, for some parameter of interest in each area. For ease of notation, we write $\theta := (\theta_1, \ldots, \theta_m)^\top$ when referring to the full vector.

Let $y_i$ denote the direct estimator of $\theta_i$ and let $d_i$ denote its sampling variance. Direct estimators are based only on area-specific data and acknowledge the sampling design by weighting the individual responses. They are unbiased under the sampling design.

For areas where the sample size is small, direct estimators can have unreasonably high variances, which may necessitate the use of model-based estimators. The foundational area-level model is the Fay--Herriot model  \citep{fay_estimates_1979}  that models the small area mean with  $y_i \ind \N{\theta_i}{d_i}$,
for areas $i=1, \ldots, m$, and $\theta_i$ is further modeled as a linear function of covariates and random effects.

The Gaussian assumption for $y_i$ is supported by the design-based Central Limit Theorem \citep{hajek_asymptotic_1964}. The assumption that $d_i$ is known is more unusual. In most statistical settings, variances must be estimated. But in survey sampling, design-based variance estimators such as Taylor linearization or replication methods provide $d_i$ directly from the sampling design, independent of any model for $\theta$ \citep{lohr_sampling_1999}. Thus, the common assumption in area-level modeling is that the variances are known. In our work, we treat the $\theta$ as fixed finite-population means and evaluate model-based estimators by averaging over sampling and thinning randomness, without subscribing to the model's assumption that $\theta$ is random. We formalize this framework in Section~\ref{sec:mse_validation}.

\subsection{Gaussian Data Thinning} \label{sec:dt_background}

The concept of splitting a single Gaussian observation for training
and validation has been explored in several recent works \citep{rasines_splitting_2023, oliveira_unbiased_2024, leiner_data_2025}.
\citet{neufeld_data_2024} unified this approach for the full class
of convolution-closed distributions through data thinning.
In the Gaussian case, data thinning decomposes observation $y_i \sim \N{\theta_i}{d_i}$ into training and test sets $\ytrain_i$ and $\ytest_i$ such that: (i) the two parts sum to the original observation, $\ytrain_i + \ytest_i = y_i$; (ii) the two parts are marginally independent, $\ytrain_i \indep \ytest_i$; and (iii) both components follow Gaussian distributions with known parameters. Remarkably, all three properties can be achieved simultaneously via Algorithm \ref{alg:thin}. The resulting marginal distributions are $\ytrain_i \sim \N{\epsilon \theta_i}{\epsilon d_i}$ and $\ytest_i \sim \N{(1-\epsilon)\theta_i}{(1-\epsilon)d_i}$.

\begin{algorithm}[H]
\caption{Gaussian Data Thinning (Algorithm 1 of \citet{neufeld_data_2024})}\label{alg:thin}
\begin{algorithmic}[1]
\Require Direct estimate $y_i \sim \N{\theta_i}{d_i}$ with known variance $d_i$
\Require Thinning parameter $\epsilon \in (0, 1)$
\State Draw $\ytrain_i \mid y_i \sim \N{\epsilon y_i}{\epsilon(1-\epsilon)d_i}$
\State Set $\ytest_i = y_i - \ytrain_i$
\State \Return Training observation $\ytrain_i$ and test observation $\ytest_i$
\end{algorithmic}
\end{algorithm}

It is worth distinguishing how data thinning creates independence. Conditional on the observed data $y_i$, the training and test components are perfectly negatively correlated since $\ytest_i = y_i - \ytrain_i$. Independence and the out-of-sample validation enabled by data thinning holds only \emph{marginally} without conditioning on the specific realization of $y_i$. The key implication is that error estimates derived from data thinning are unbiased only in expectation over hypothetical sample datasets, not for the error on any particular dataset. This fundamental limitation, where validation does not target dataset-specific performance, also arises in cross-validation  \citep{bates_cross-validation_2024}.

For the algorithm above to actually produce marginally independent observations, two conditions must hold: $y_i$ must be Gaussian and the variance must be known. Proposition 10 of \citet{neufeld_data_2024} shows that if thinning is performed using an incorrect variance $\tilde{d}_i$ instead of the true $d_i$, the resulting sets have covariance
\[
\Cov{\ytrain_i}{\ytest_i} = \epsilon(1-\epsilon)(d_i - \tilde{d}_i).
\]
Thus underestimating the variance ($\tilde{d}_i < d_i$) induces positive correlation between sets, while overestimating ($\tilde{d}_i > d_i$) induces negative correlation. In practice, the design-based variance estimators used in survey sampling are generally reliable, but practitioners should be aware that substantial misspecification of $d_i$ will compromise the data thinning approach.

\subsection{Motivating Example: Selecting Spatial Basis Functions} \label{sec:motivating_ex}

We now turn to a concrete model selection problem that motivates our theoretical analysis. Spatial correlation is common in SAE. Neighboring regions often share economic conditions, demographic composition, or policy environments. Capturing this structure requires model choices: how much spatial smoothing is appropriate and how can we validate this choice?

Consider modeling with spatial basis functions, commonly used due to reduced computational burden compared to other spatial models. Following \citet{hughes_dimension_2013}, we construct basis functions using the Moran operator \citep{moran_notes_1950}, which captures spatial autocorrelation orthogonal to an initial covariate matrix based on the adjacency structure of the data (construction details in Section~\ref{sec:empirical}). We use the $p$ leading eigenvectors as covariates in the Fay--Herriot model to capture spatial dependence across areas. Small values of $p$ result in strong spatial smoothing while higher values of $p$ result in less. Figure~\ref{fig:map} illustrates this progression for the spatial effects for a sample dataset in California, created using the American Community Survey Public Use Microdata Sample (PUMS). Using $p = 6$ basis functions results in much more spatial smoothing, while the model with $p = 42$ shows finer local variation.

\begin{figure}[H]
  \centering
  \includegraphics[width=0.9\linewidth]{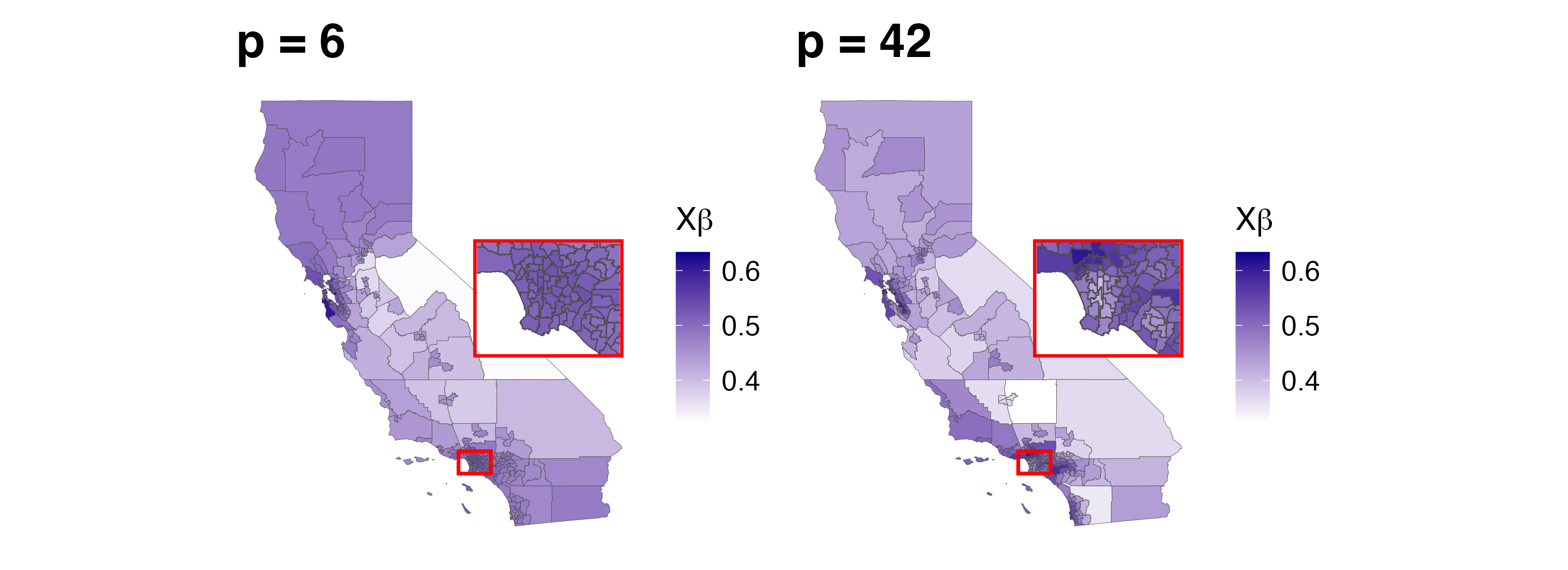}
\caption{Spatial covariate effects for the Fay--Herriot model for example data created using PUMS for California. Using $p = 6$ basis functions results in much more spatial smoothing. The model with $p = 42$ shows much finer local variation, particularly in the north and the southern regions of the state including Greater Los Angeles, shown in the zoomed-in rectangle. We use this as our empirical model validation example in subsequent sections.}
  \label{fig:map}
\end{figure}

\newpage
Currently, there is no clear way to determine how many basis functions should be selected.  \citet[p.~156]{hughes_dimension_2013} suggest using roughly 10\% of available eigenvectors, a heuristic many in the field adopt directly, though they note that a DIC-based approach ``is obviously more defensible.'' More recently, \citet{janicki_bayesian_2022} observe that ``the choice of the number of basis functions to use in the model specification remains an open question''.

The PUMS data used in Figure~\ref{fig:map} provides access to a complete set of microdata, so we can treat the population quantities as known and then subsample. Thus, this setup allows for design-based simulations ideal for studying model selection methods with a real practical need. We use this example throughout Section~\ref{sec:mse_validation} to illustrate theoretical results. This example also forms the foundation for our empirical analysis and model comparison in Section~\ref{sec:empirical}, where we provide full details on the sample generation and spatial basis functions.


\section{MSE-Based Validation with Data Thinning}\label{sec:mse_validation}

\subsection{Assumptions and Conventions}


\begin{assumption}[Finite-population framework with known variances]\label{assump:framework}
We assume: (i) $y_i \ind N(\theta_i, d_i)$ for $i=1, \ldots, m$, where $\theta_i$ are fixed unknown finite-population means; (ii) the sampling variances $d_i$ are known.
\end{assumption}

Model-based estimators are evaluated under this finite-population perspective: we assess accuracy using only the assumptions above. While we assess estimates from Bayesian models that treat $\theta$ as random in Section~\ref{sec:empirical}, the data thinning procedure and our theoretical results require only Assumption~\ref{assump:framework}.
It is important to distinguish between the randomness from the sampling process producing $y$, and the thinning procedure producing $\{\ytrain, \ytest\}$. We use subscripts on expectation operators to clarify which source of variability is being averaged over: $\E[y]{\cdot}$ for the sampling distribution, $\E[\ytrain,\ytest]{\cdot}$ for the marginal distribution of thinned data (unconditional on $y$), and $\E[\ytrain,\ytest]{\cdot \given y}$ when conditioning on the realized dataset and averaging only over thinning. When $\E{\cdot}$ appears without subscript, the relevant source of randomness will be clear from context or stated explicitly.


\subsection{MSE Estimation using Data Thinning}
We consider model comparison based on their expected squared error averaged across all $m$ areas, conditional on the fixed means $\theta$. We refer to this quantity as the mean squared error (MSE).  It differs from area-specific expected squared errors sometimes used in the SAE literature \citep[see e.g.,][]{prasad_estimation_1990}. We focus on the aggregate because practitioners often use model selection to improve overall performance across all areas. Our natural target estimand for model comparison is
\[
\msefull = \frac{1}{m}\sum_{i=1}^m \E[y]{(\hat{\theta}_i - \theta_i)^2},
\]
where the expectation is taken over the sampling distribution of the full data $y$. Since this requires knowledge of $\theta_i$, we refer to this quantity as the full-data oracle MSE.  It measures the expected performance of the estimates practitioners would actually deploy, computed using the full data.

Consider the data thinning setup where we split $y_i \sim \N{\theta_i}{d_i}$ into marginally independent $\ytrain_i$ and $\ytest_i$ following Algorithm \ref{alg:thin}. In the context of SAE, we can treat the training and test observations as new replicate direct estimates of $\theta_i$ after scaling, with effective sample sizes of $\epsilon n_i$ and $(1-\epsilon)n_i$ respectively, i.e.,
\[
 \ytrain_i / \epsilon \ind \N{\theta_i}{d_i/\epsilon}, \qquad  \ytest_i /(1-\epsilon) \ind \N{\theta_i}{d_i/(1-\epsilon)}.
\]
Figure~\ref{fig:thinning_maps} in the Appendix visualizes this for a single sample.

The thinning parameter $\epsilon$ controls the allocation of information across sets and the relative variances of these new direct estimates. At a high level, model validation can be carried out by fitting a model using $\ytrain/\epsilon$ to create estimates of $\theta$, and then the MSE can be evaluated on $\ytest/(1-\epsilon)$. We define the thinned-data oracle MSE as the expected squared error of the procedure trained on $\epsilon$-thinned data, i.e.,
\[
\msethin := \frac{1}{m}\sum_{i=1}^m \E[\ytrain]{\left( \thetatrain_i - \theta_i \right)^2},
\]
where $\thetatrain$ are estimates of $\theta$ created from $\ytrain$, and the expectation is taken over the marginal distribution of $\ytrain$, which includes randomness of the data $y$ and the randomness from the thinning procedure. Note that $\msethin$ reduces to $\msefull$ when $\epsilon = 1$. But setting $\epsilon = 1$ dedicates all information to estimation and leaves nothing for validation, making $\msefull$, the ideal target, unachievable unless $\theta$ is known.

Data thinning allows us to estimate $\msethin$ with $\ytrain_i$ and $\ytest_i$. We propose the following estimator for a single set of thinned data,
\[
\msehat := \frac{1}{m}\sum_{i=1}^m \left[ \left( \thetatrain_i - \tfrac{1}{1-\epsilon}\ytest_i \right)^{\!2} - \tfrac{d_i}{1-\epsilon} \right].
\]
The term $d_i/(1-\epsilon)$ corrects for bias from using $\ytest_i /(1-\epsilon)$ as a surrogate for $\theta_i$; this correction can yield negative values if test noise exceeds the squared error.
\begin{theorem}[Unbiased MSE estimation]
\label{thm:unbiased_mse}
The estimator $\msehat$ is unbiased for the thinned-data oracle MSE:
\[
\E{\msehat} = \msethin,
\]
where the expectation is taken over the joint distribution of $(y^{(1)}, y^{(2)})$, unconditional on $y$.
\end{theorem}

\begin{proof}
We take the expectation over the joint distribution of $\ytrain$ and $\ytest$ to get
\begin{align*}
    \E{\msehat}
& = \frac{1}{m}\sum_{i=1}^m
   \E[\ytrain, \ytest]{ \left( \thetatrain_i - \tfrac{1}{1-\epsilon}\ytest_i \right)^{\!2}
         - \tfrac{d_i}{1-\epsilon}} \\
&=  \frac{1}{m}\sum_{i=1}^m\E[\ytrain]{\E[\ytest]{\left( \thetatrain_i - \tfrac{1}{1-\epsilon}\ytest_i \right)^{\!2}} - \tfrac{d_i}{1-\epsilon}},
\end{align*}
where the second equality results from the marginal independence of $\ytrain$ and $\ytest$, unconditional on $y$.

Fix an area $i$ and define
\[
\delta_i := \thetatrain_i - \theta_i,
\qquad
\eta_i := \tfrac{1}{1-\epsilon}\ytest_i - \theta_i .
\]
Note that $\delta_i$ is only random with respect to $\ytrain$, while $\eta_i \ind N(0, d_i/(1-\epsilon))$ is random with respect to $\ytest$.

Thus we have
\[
\E[\ytest]{(\delta_i-\eta_i)^2}
=
\delta_i^2 - 2 \delta_i \cdot \E[\ytest]{\eta_i} +  \E[\ytest]{\eta_i^2}
=
\delta_i^2 + \tfrac{d_i}{1-\epsilon}.
\]
Substituting into the definition of $\msehat$ and taking $\E[\ytrain]{\cdot}$ gives
\begin{align*}
    \E{\msehat}
&=  \frac{1}{m}\sum_{i=1}^m\E[\ytrain]{\E[\ytest]{(\delta_i-\eta_i)^2} - \tfrac{d_i}{1-\epsilon}}\\
&=  \frac{1}{m}\sum_{i=1}^m \E[\ytrain]{\delta_i^2 + \tfrac{d_i}{1-\epsilon} - \tfrac{d_i}{1-\epsilon}}  \\
&=
\frac{1}{m}\sum_{i=1}^m \E[\ytrain]{\left( \thetatrain_i - \theta_i \right)^2}
=
\msethin.
\end{align*}
\end{proof}

A natural question is whether $\msehat$ can serve as a proxy for the target quantity, $\msefull$. The following decomposition provides a useful framework for understanding the inherent tension in model validation using data thinning:
\begin{equation}\label{eq:tradeoff}
        \E{(\msehat - \msefull)^2} = \bigg(\, \underbrace{\msethin - \msefull}_{\text{The Thinning Gap}}\, \bigg)^2 +
        \underbrace{\Var{\msehat}}_{\text{The Estimator Variance}}.
\end{equation}
Note that $\msehat$ is the only random quantity in this decomposition. Both $\msethin$ and $\msefull$ are constants given $\epsilon$. Since $\msehat$ is unbiased for $\msethin$, the cross-term vanishes. Therefore, to estimate $\msefull$ for model comparison using data thinning, we must account for both the thinning gap between $\msethin$ and $\msefull$ and the variability of $\msehat$. The balance of these quantities changes with the thinning parameter $\epsilon$. We examine them in the next two subsections.


\subsection{The Thinning Gap}\label{sec:thinning_gap}

The thinning gap, $\msethin - \msefull$, is the systematic difference between the thinned-data and full-data MSE for predicting the fixed finite-population mean $\theta_i$ from Assumption~\ref{assump:framework}. Since $\theta_i$ is common to all candidate models, this difference is comparable across model choices. To gain insight into its structure, we derive this gap analytically under a correctly specified Fay--Herriot model. We first review how shrinkage arises in the classical formulation,
\[
y_i \ind \N{\theta_i}{d_i}, \quad \theta_i = x_i^\top \beta + u_i, \quad u_i \iid N(0, \sigma^2),
\]
where $x_i$ are $p$-dimensional covariate vectors, $\beta$ is the coefficient vector, and $u_i$ are random effects capturing residual variability in $\theta$.

Assuming $\beta$ and $\sigma^2$ are known, the posterior mean of $\theta_i$ given $y_i$ is
\[
\tilde{\theta}_i = \gamma_i \, y_i + (1 - \gamma_i) \, x_i^\top \beta, \qquad \gamma_i = \frac{\sigma^2}{\sigma^2 + d_i},
\]
where $\gamma_i \in (0,1)$ governs the balance between the direct estimate $y_i$ and the regression prediction $x_i^\top\beta$. When the sampling variance $d_i$ is large relative to $\sigma^2$, the direct estimate is unreliable and $\gamma_i$ is small, so $\tilde{\theta}_i$ shrinks toward the regression term that borrows strength across all areas. When $d_i$ is small, the direct estimate dominates. Model complexity affects this balance: more flexible models can explain more variation in $\theta$ through the mean structure, reducing $\sigma^2$ and shifting the estimator toward the model-based component.

For thinned data, we work with the rescaled direct estimate $\tfrac{1}{\epsilon} y_i^{(1)} \sim \N{\theta_i}{d_i/\epsilon}$, which has inflated variance. The corresponding posterior mean is
\[
\tilde{\theta}_i^{(1)} = \gamma_i(\epsilon) \, \tfrac{1}{\epsilon} y_i^{(1)} + (1 - \gamma_i(\epsilon)) \, x_i^\top \beta, \qquad \gamma_i(\epsilon) = \frac{\sigma^2}{\sigma^2 + d_i/\epsilon}.
\]
Since thinning inflates the sampling variance from $d_i$ to $d_i/\epsilon$, the thinned estimator shrinks more toward the regression term since $\gamma_i(\epsilon) < \gamma_i$ for all $\epsilon < 1$.

We now show that the expected thinning gap is positive and its magnitude depends on the shrinkage behavior of the model.

\begin{proposition}
\label{prop:thinning_gap}
(MSE thinning gap under known parameters) Under the correctly specified Fay--Herriot model with known $\beta, \sigma^2$, $$\msethin - \msefull = \frac{1}{m} \sum_{i=1}^m \Delta_i(\epsilon)$$
where
\[\Delta_i(\epsilon) = \frac{1-\epsilon}{\epsilon} \cdot \gamma_i(\epsilon) \, \gamma_i \, d_i > 0.\]
\end{proposition}

\begin{proof}
Assuming correct model specification, the expected squared error for the full-data case follows from \citet{prasad_estimation_1990}. The thinned-data case follows identically with inflated variance:
\[
\E{(\tilde{\theta}_i - \theta_i)^2}= \gamma_i d_i, \qquad \E{(\tilde{\theta}_i^{(1)} - \theta_i)^2} = \gamma_i(\epsilon) \cdot d_i/\epsilon.
\]
The per-area gap is therefore
\[\Delta_i(\epsilon) = \frac{\gamma_i(\epsilon) d_i}{\epsilon} - \gamma_i d_i = d_i \left( \frac{\gamma_i(\epsilon)}{\epsilon} - \gamma_i \right).\]

We simplify the term inside the parenthesis to get
\begin{align*}
    \frac{\gamma_i(\epsilon)}{\epsilon} - \gamma_i &= \frac{\sigma^2}{\sigma^2 \epsilon + d_i} - \frac{\sigma^2}{\sigma^2 + d_i} = \frac{\sigma^4 (1-\epsilon)}{( \epsilon  \sigma^2  + d_i)(\sigma^2 + d_i)}
\end{align*}
which is positive for all $\epsilon \in (0,1)$.
This term further simplifies to
\[
= \frac{1-\epsilon}{\epsilon}  \cdot \frac{\sigma^2}{\sigma^2 + d_i}  \cdot\frac{ \sigma^2}{\sigma^2  + d_i/\epsilon}
\]
Substituting the definitions of $\gamma_i(\epsilon)$ and $\gamma_i$, we have the result
$$\Delta_i(\epsilon) = \frac{1-\epsilon}{\epsilon} \cdot \gamma_i(\epsilon) \, \gamma_i \, d_i.$$
\end{proof}

\begin{remark}
Since $\gamma_i(\epsilon) < \gamma_i$ for all $\epsilon < 1$, the per-area gap satisfies
\[
\Delta_i(\epsilon) < \frac{1-\epsilon}{\epsilon} \, \gamma_i^2 \, d_i.
\]
This upper bound depends only on the full-data shrinkage factor $\gamma_i$, providing a simple diagnostic. Practitioners can fit each model on full data to compute $\gamma_i$ and compare bounds without evaluating each candidate model at multiple values of $\epsilon$.
\end{remark}

The thinning gap vanishes as $\epsilon$ approaches 1, in which case $\tilde{\theta}_i^{(1)}$ and $\tilde{\theta}_i$ are estimates from nearly identical data. More critically, the thinning gap is model-dependent. Different candidate models imply different shrinkage levels $\gamma_i$, and hence different gaps. Models relying less on random effects, i.e., with smaller $\sigma^2$, incur smaller gaps, while models with stronger random effects systematically appear worse under thinned-data validation than they would perform on full data. The implications of this model-dependent gap for model selection are examined in Section~\ref{sec:fundamental_tension}.

\begin{figure}[htbp]
  \centering
  \includegraphics[width=0.9\linewidth]{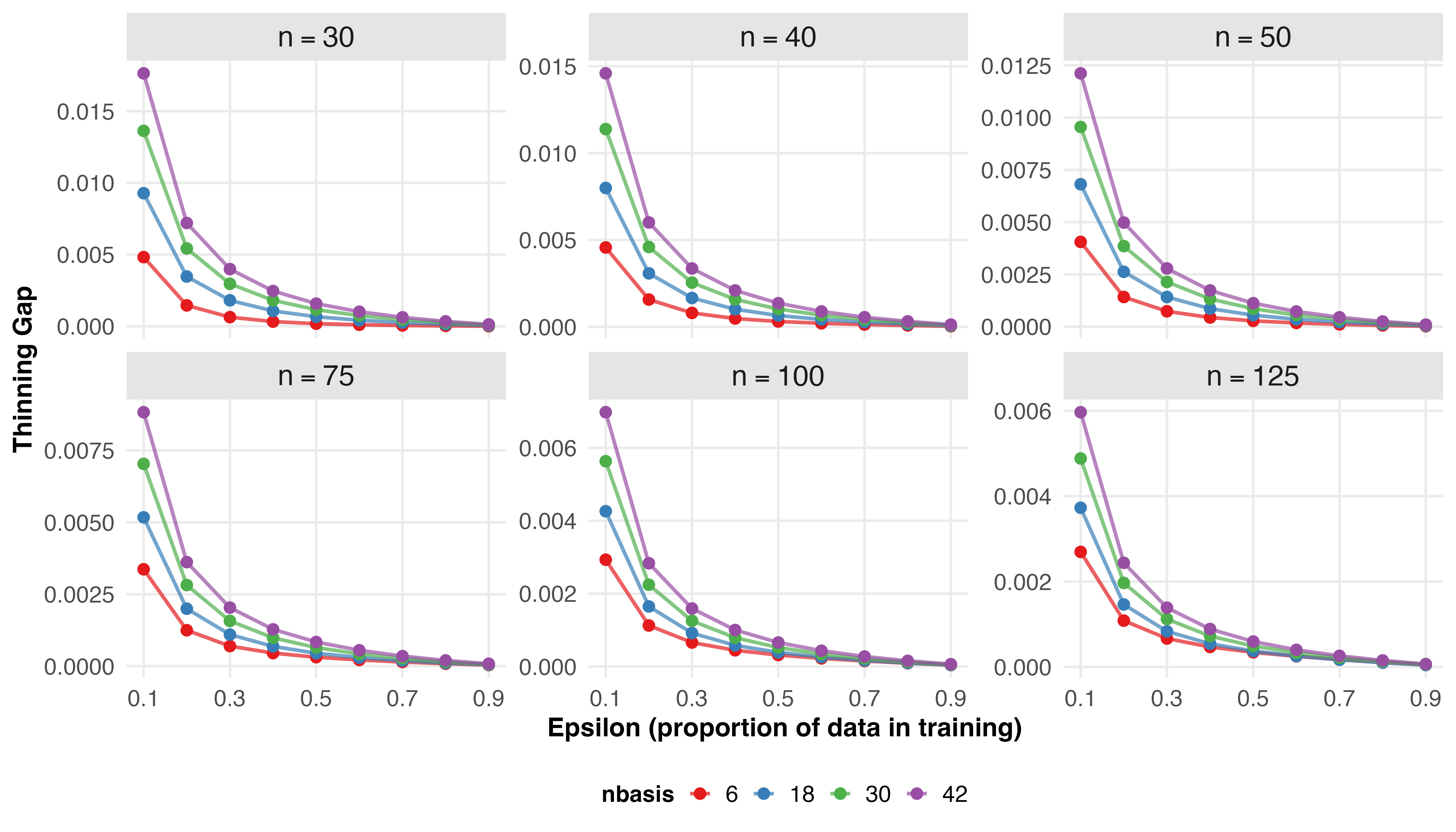}
\caption{Average realized thinning gap for Fay--Herriot models with $p = 6, 18, 30, 42$ spatial
basis functions, averaged over 50 independent samples. Each panel corresponds to an equal allocation design with the indicated target $n$. Complex models (higher $p$) exhibit larger gaps, particularly at low $\epsilon$.}
  \label{fig:systematic-gap}
\end{figure}

Proposition \ref{prop:thinning_gap} assumes known parameters. In practice, parameters must be estimated, introducing an additional effect on the thinning gap that varies with model complexity. Consider the special case of the Fay--Herriot when $\sigma^2$ is known but where $\beta$ must be estimated. Under this setting, the Best Linear Unbiased Predictor (BLUP) is
\[
\tilde{\theta}_i = \gamma_i \, y_i + (1-\gamma_i) \, x_i^\top \tilde{\beta}
\]
where $\tilde{\beta}$ is the weighted least-squares estimator.  We extend Proposition \ref{prop:thinning_gap} for this BLUP case.

\begin{proposition}
\label{prop:thinning_gap2}
(MSE thinning gap under estimated $\beta$) Under the correctly specified Fay--Herriot model with known $\sigma^2$ but estimated $\beta$,
\[\msethin - \msefull = \frac{1}{m} \sum_{i=1}^m \Delta_i(\epsilon),\]
where
$$\Delta_i(\epsilon) = \frac{1-\epsilon}{\epsilon} \cdot \gamma_i(\epsilon) \, \gamma_i \, d_i + \big[g_{2i}(\epsilon) - g_{2i}\big],$$
with
$$g_{2i}(\epsilon) = (1-\gamma_i(\epsilon))^2 \cdot x_i^\top \left[\sum_{j=1}^m \frac{x_j x_j^\top}{\sigma^2 + d_j/\epsilon}\right]^{-1} x_i.$$
and $g_{2i} := g_{2i}(\epsilon=1)$ denoting the full-data case.

Under the intercept-only model this simplifies to
$$\Delta_i(\epsilon) = \frac{1-\epsilon}{\epsilon} \cdot \gamma_i(\epsilon) \, \gamma_i \, d_i + \left[\frac{(1-\gamma_i(\epsilon))^2}{w(\epsilon)} - \frac{(1-\gamma_i)^2}{w}\right].$$
where $w(\epsilon) = \sum_{j=1}^m (\sigma^2 + d_j/\epsilon)^{-1}$ and $w := w(\epsilon=1)$, again denoting the full-data case.
\end{proposition}

\begin{corollary}
\label{cor:thinning_gap_monotone}
(Term-wise monotonicity of the thinning gap) Under the conditions of Proposition~\ref{prop:thinning_gap2} and a full-rank covariate matrix $X$, the systematic gap $\frac{1}{m}\sum_{i=1}^m \tfrac{1-\epsilon}{\epsilon}\gamma_i(\epsilon)\gamma_i d_i$ and the parameter-uncertainty gap $\frac{1}{m}\sum_{i=1}^m [g_{2i}(\epsilon) - g_{2i}]$ are each strictly decreasing in $\epsilon$ on $(0,1]$. In particular, the total thinning gap $\msethin - \msefull$ is strictly decreasing in $\epsilon$.
\end{corollary}
Proofs of Proposition~\ref{prop:thinning_gap2} and Corollary~\ref{cor:thinning_gap_monotone} are given in Appendices~\ref{app:thinning_gap2} and~\ref{app:thinning_gap_monotone}, respectively.

When $\beta$ is estimated, the thinning gap acquires an additional term reflecting uncertainty in $\beta$ (Proposition~\ref{prop:thinning_gap2}). Both terms are strictly decreasing in $\epsilon$ (Corollary~\ref{cor:thinning_gap_monotone}) but they respond differently to model complexity. Models with stronger shrinkage have a smaller known-parameter gap, but place more weight on the regression component $x_i^\top\tilde{\beta}$, inflating the $(1-\gamma_i)^2$ and $(1-\gamma_i(\epsilon))^2$ factors and making the estimator more sensitive to parameter estimation error. These effects oppose each other, and their net balance is not analytically straightforward. Proposition~\ref{prop:g2i_covariates} in the Appendix shows that, for nested models with fixed $\sigma^2$, the error from estimating the regression component is non-decreasing in the number of covariates $p$.

Figure~\ref{fig:systematic-gap} illustrates both results empirically using Fay--Herriot models with $p = 6, 18, 30, 42$ spatial basis functions across six equal allocation survey designs, where Poisson sampling yields expected within-area sample size $\E{n_i}=n$ for all areas $i$ (see Section~\ref{sec:sim_framework} for details). The figure shows the \textit{realized} thinning gap averaged across 50 samples. The monotonic decay in $\epsilon$ confirms the corollary. More complex models (higher $p$) exhibit larger gaps at any given $\epsilon$, consistent with Proposition~\ref{prop:g2i_covariates} and suggesting that the parameter estimation cost dominates the shrinkage benefit for this particular setting. This difference is most pronounced at low $\epsilon$ and shrinks as $\epsilon$ approaches 1.


\subsection{The Estimator Variance}\label{sec:variance}

We now turn to the second term in the decomposition of \eqref{eq:tradeoff}, the variance of the MSE estimator. Proposition~\ref{prop:mse_variance} provides further decomposition of this variance.

\begin{proposition}[Variance of the MSE estimator]
\label{prop:mse_variance}
The variance of the MSE estimator over thinning splits is
\begin{align}
    &\Varthin{\msehat}  \nonumber \\
    &= \E[\ytrain]{ \Var[\ytest]{ \msehat \agiven \ytrain} }
 + \Var[\ytrain]{ \E[\ytest]{\msehat \agiven \ytrain } }  \label{eq:var_lotv} \\
  &=  \frac{2}{m^2}\sum_{i=1}^m \left(
\left(\tfrac{d_i}{1-\epsilon}\right)^{\!2} + 2\,\tfrac{d_i}{1-\epsilon}\,\E[\ytrain]{(\thetatrain_i - \theta_i)^2}
\right) + \Var[\ytrain]{\frac{1}{m}\sum_{i=1}^m \left( \thetatrain_i - \theta_i \right)^2}.
\label{eq:var_closedform}
\end{align}
\end{proposition}
The full derivation is provided in Appendix~\ref{app:var_deriv}. The variance decomposes into two components via the law of total variance. The first is the \emph{test-set variability}, capturing randomness from the held-out $\ytest$ given fixed training data. The second is the \emph{training-set variability}, capturing fluctuation in estimates across different training splits. In the closed form~\eqref{eq:var_closedform}, the summation corresponds to test-set variability while the final variance term corresponds to training-set variability. The test-set component decomposes further into a purely test-driven term and an interaction term that is amplified for areas with larger training errors. Note that the training-set variability here is the variability of the thinned-data oracle MSE across different realizations of $\ytrain$.

\begin{figure}[htbp]
    \centering
    \includegraphics[width=0.9\linewidth]{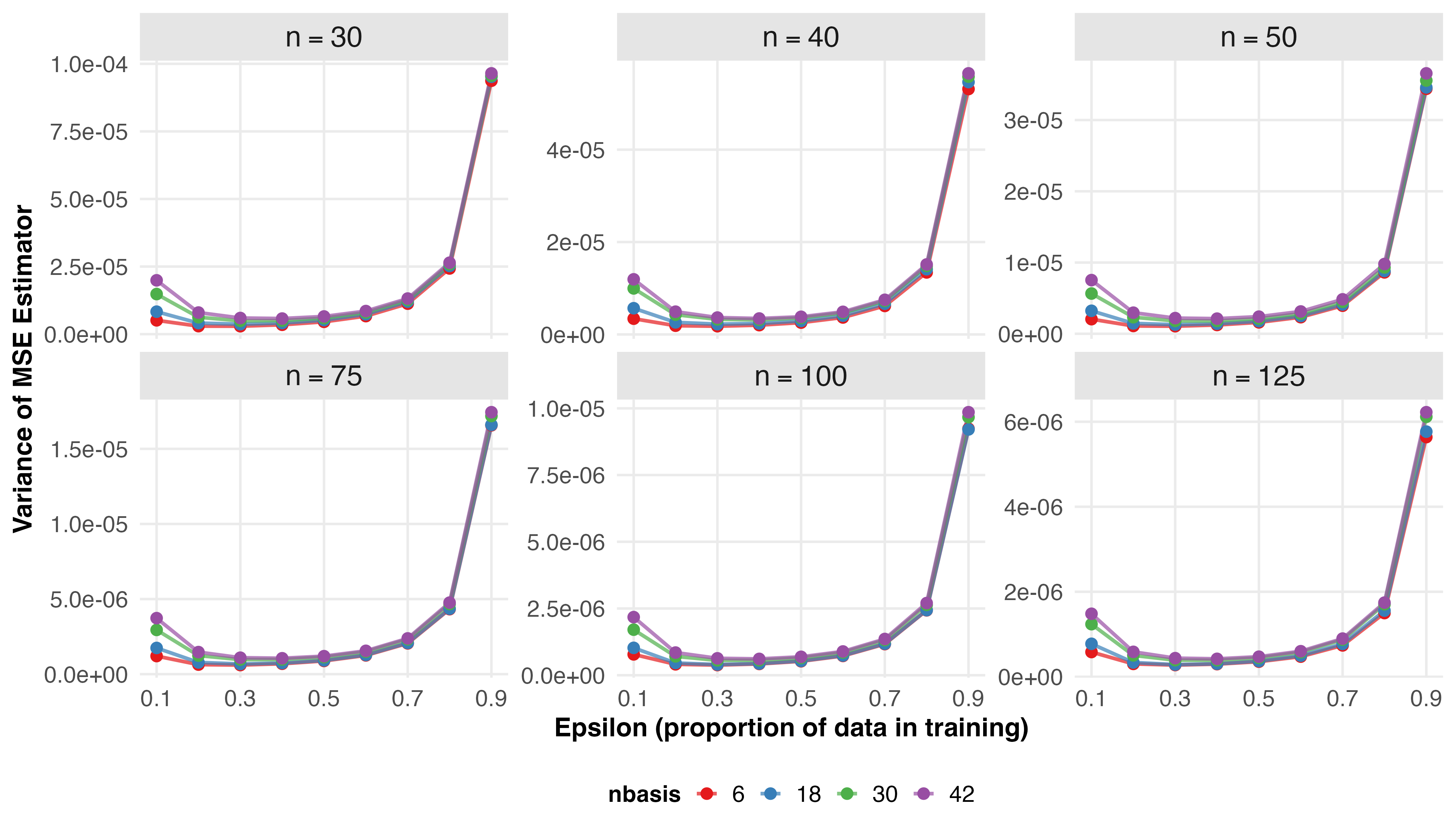}
    \caption{Variance of the MSE estimator for Fay--Herriot models with $p = 6, 18, 30, 42$ spatial basis functions, computed across 50 independent samples. Each panel corresponds to an equal allocation survey design with the indicated sample size per area. The variance is minimized  at $\epsilon \approx 0.3$--$0.4$, with notable increases for $\epsilon \geq 0.8$.}
    \label{fig:dt-variance}
\end{figure}

Similar to the thinning gap, the estimator variance is model-dependent. To understand the behavior of the variance, it is useful to consider the special case where we use the direct estimator $\hat{\theta}_i^{(1)} := \ytrain_i/\epsilon$ as an estimate of $\theta_i$. In this case, the variance simplifies to
\[
\Varthin{\msehat} = \frac{2}{m^2}\sum_{i=1}^m \frac{d_i^2}{\epsilon^2(1-\epsilon)^2},
\]
which is minimized at $\epsilon = 1/2$ (see Appendix~\ref{app:direct_var_opt}). When no shrinkage is involved, the test-set and training-set variability balance exactly at $\epsilon = 1/2$. This provides a baseline which helps us understand the following results.

\begin{proposition}[Variance-minimizing $\epsilon$ for the Fay--Herriot model with known parameters]
\label{prop:fh_var_minimize}
Under the Fay--Herriot model with known $\beta$ and $\sigma^2$, the variance is minimized at some $\epsilon^* \in (0, 1/2)$ and strictly increasing for $\epsilon \in [1/2, 1)$.

Moreover, the area-specific variance contribution is minimized at
\[
\epsilon_i^* \;=\; \max\!\left\{\,0,\; \frac{1}{2} - \frac{d_i}{2\sigma^2}\right\}.
\]
\end{proposition}

See Appendix~\ref{app:fh_var_minimize} for the proof of these results. Compared to the direct estimator, shrinkage estimators benefit less from additional training data due to the borrowing of strength across areas. This is made clear in the area-specific result, where optimum equals $1/2$ adjusted downward by $d_i/2\sigma^2$ which is one-half the ratio of the variance of the direct estimator and the regression model. For areas with weak shrinkage where $\sigma^2\gg d_i$, we have $\epsilon_i^* \approx 1/2$, which approaches the direct estimator baseline. For strong-shrinkage areas, where $\sigma^2 \leq d_i$, the optimum is zero, since the estimate for that area does not improve by incorporating the direct estimate if parameters are known.  However, in practice, the optimal value would not be zero since the parameter estimation necessitates pooling information across all areas.

Figure~\ref{fig:dt-variance} shows the empirical variance of the MSE estimator across $\epsilon$ and demonstrates that the theoretical properties hold. The variance is minimized at $\epsilon \approx 0.3$--$0.4$. The asymmetry is also clear: the variance spikes substantially for $\epsilon \geq 0.8$ as the test set shrinks. Complex models exhibit uniformly higher variance than simple models, though this difference is much less pronounced compared to the thinning gap for moderate $\epsilon>0.3$.

Crucially, these results establish that the variance of the MSE estimator is strictly increasing for the Fay--Herriot model (as well as the direct estimator) as $\epsilon$ approaches 1. This directly opposes the thinning gap, which strictly decreases in $\epsilon$ over $[0, 1]$.  The two terms in decomposition~\eqref{eq:tradeoff} pull in opposite directions.


\subsection{The Bias-Variance Tradeoff for Data Thinning}\label{sec:fundamental_tension}

The thinning gap and estimator variance results of the previous two subsections reveal competing demands on $\epsilon$. The MSE estimator, $\msehat$, is unbiased for the thinned-data target $\msethin$, but what we actually want is the full-data target $\msefull$. Closing this gap requires high $\epsilon$, yet high $\epsilon$ inflates variance of $\msehat$ as the test set shrinks. Moreover, this trade-off is model-dependent: model complexity and shrinkage behavior affect both the thinning gap and the variance. \textit{There is no single $\epsilon$ that is uniformly optimal across candidate models.}

This model-dependence has direct consequences for model comparison. Figure~\ref{fig:tradeoff} shows the sum of squared thinning gap and variance of the MSE estimator averaged across 50 samples from six equal allocation designs (as in Figures~\ref{fig:systematic-gap} and \ref{fig:dt-variance}). At low $\epsilon$, the curves are widely separated across models. The thinning gap dominates and amplifies between-model differences, systematically favoring simpler models. As $\epsilon$ increases past $0.5$, the curves converge and become relatively flat through $\epsilon \approx 0.7$. The per-model optimum differs ($\epsilon \approx 0.4$--$0.6$, increasing with model complexity), but the convergence at moderate $\epsilon$ means that all candidate models have similar estimation errors in their MSE estimates. This is the key property for fair model comparison.

\begin{figure}[htbp]
    \centering
  \includegraphics[width=0.9\linewidth]{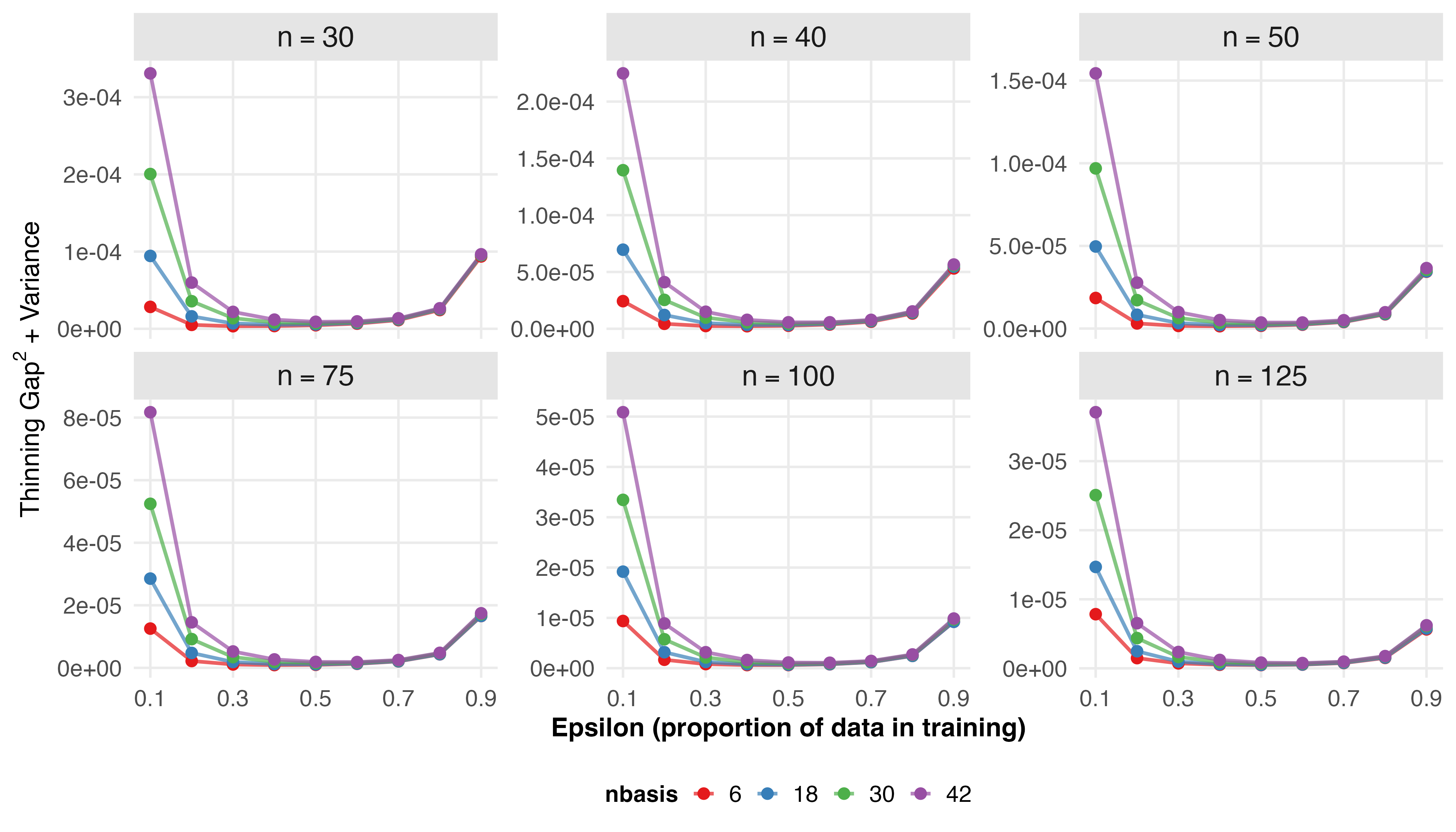}
    \caption{The thinning gap-variance trade-off for Fay--Herriot models with $p = 6, 18, 30, 42$ spatial basis functions. Curves show the sum of squared thinning gap and variance of the MSE estimator averaged across 50 samples from each design. The curves are relatively flat for $\epsilon$ between $0.4$ to $0.7$ across different designs. A log-scale version of the same plot is shown in Appendix~\ref{app:log_tradeoff} which is more helpful to see the differing optima per model and where the between-model differences in curves shrink for $\epsilon>0.5$.}
    \label{fig:tradeoff}
  \end{figure}

In principle, an unbiased estimator of $\msefull$ can be constructed by evaluating $\msehat$ at multiple $\epsilon$ values and extrapolating to $\epsilon = 1$. However, the $\epsilon$ values closest to $\epsilon = 1$, which are most informative for the extrapolation as they are nearest the target $t = 0$, are precisely where $\Var{\msehat}$ is largest. Thus, in all settings we examined, the extrapolation estimator underperformed a single $\epsilon$, suggesting that the variance introduced by extrapolation outweighs the bias it removes (Appendix~\ref{app:extrapolation}).

Given the bias-variance tradeoff and the limitations of extrapolation, we recommend a single $\epsilon \in [0.5, 0.7]$ for model comparison. Within this range, the thinning gap is moderate and similar across candidate models (\S\ref{sec:thinning_gap}), estimation variance has not yet entered the rapid-growth regime near $\epsilon = 1$ (\S\ref{sec:variance}), and the tradeoff curves in Figure~\ref{fig:tradeoff} are relatively flat, making the comparison robust to the precise choice of $\epsilon$. We examine the empirical performance of this recommendation in Section~\ref{sec:empirical}.


\section{Repeated Thinning}\label{sec:repeated}

The MSE estimator described in Section \ref{sec:mse_validation} applies the data thinning procedure only once. A natural question is whether averaging the MSE estimators over multiple thinning procedures can improve error estimates. We describe and compare two such approaches that are currently available.


The first approach to reduce the randomness from a single data thinning procedure is what we call \textit{repeated single-fold thinning}.
This approach simply repeats data thinning $R$ times at a fixed $\epsilon$ and average the $R$ separate errors, as done in the application by \citet{dharamshi_generalized_2025}. An alternative is the so-called \textit{multi-fold thinning}, which splits the data into $K > 2$ components that, marginally, are mutually independent \citep{neufeld_data_2024}. Mirroring $K$-fold cross-validation, the $k$th component is used as the test data while training proceeds on the aggregated remainder. As proposed by the authors, the training set uses $K-1$ of the $K$ components; this corresponds to training fraction $\epsilon = (K-1)/K$. Averaging over all $K$ sets yields the multi-fold error estimate. See Appendix~\ref{app:multifold} for the full algorithm for Gaussian data.

The key difference between the two methods lies in the conditional dependence structure among training-test pairs given $y_i$. Under repeated thinning, training and test sets from distinct repeats are conditionally independent. In contrast, multi-fold thinning creates $(y_i^{(1)}, \ldots, y_i^{(K)})$ that are conditionally dependent and add up to $y_i$.  See \citet{neufeld_data_2024}, Example 5, for the full details including the joint distribution of $(y_i^{(1)}, \ldots, y_i^{(K)})$. Thus, the test sets $y^{(k)}$ are conditionally dependent across the $k=1, \ldots, K$ sets. Moreover, the training sets $y_i^{(-k)} := y_i - y_i^{(k)}$ share $K-2$ components across sets and are also dependent given $y_i$.


Recall that the concurrent goals discussed in the previous section are: (i) minimize the thinning gap between the thinned-data and full-data oracle MSE quantities and (ii) reduce the variance of the MSE estimate. For a fixed $\epsilon$ and model, the thinning gap is a fixed quantity and only the estimator variance can be reduced through repeats.

For any averaged estimator derived from $J$ thinned datasets, $\msebar := J^{-1} \sum_{j=1}^J \msehat^{(j)}$, the law of total variance gives
\[
\Var{\msebar} = \underbrace{\Var[y]{\E[\ytrain, \ytest]{\msebar \given y}}}_{\text{irreducible}} + \underbrace{\E[y]{\Var[\ytrain, \ytest]{\msebar \given y}}}_{\text{reducible}}.
\]
The inner operators condition on the full data $y$ and are taken over the thinning splits; the outer operators are taken over the sampling distribution of $y$. The \emph{irreducible} component reflects variability across different realizations of the observed data $y$: no averaging scheme can reduce this term given a single dataset. The \emph{reducible} component captures variability from the thinning procedure itself, where the inner variance is conditional on a fixed $y$.

We examine the term inside the expectation for the reducible component. For repeated thinning, conditional independence yields
\[
\Var[\ytrain, \ytest]{\msebar_{\mathrm{repeat}} \given y} = \frac{1}{R} \, \Var[\ytrain, \ytest]{\msehat \given y}.
\]
For multi-fold thinning, the shared training components induce pairwise correlation between MSE estimators from different sets. Under conditional exchangeability with common correlation $\rho(y) := \Corr[\ytrain, \ytest]{\msehat^{(k)}}{\msehat^{(j)} \given y}$ for $k \neq j$, we have:
\[
\Var[\ytrain, \ytest]{\msebar_{\mathrm{multi}} \given y} = \frac{1}{K} \, \Var[\ytrain, \ytest]{\msehat \given y} \cdot \big[1 + (K-1) \rho(y) \big].
\]


Consequently, for equal computational budgets (e.g., $R=K$), repeated splitting yields smaller conditional variance than multi-fold thinning whenever $\rho(y) \ge 0$. The sign and magnitude of $\rho(y)$ for multi-fold thinning depend on the data and the model being fit.

To build intuition for why $\rho(y)$ might often be positive, consider the structure of the MSE estimator in each fold, given by
\[
\msehat^{(k)} = \frac{1}{m}\sum_{i=1}^m \left[ \left( \hat{\theta}^{(-k)}_i - \tfrac{1}{1-\epsilon}y^{(k)}_i \right)^{\!2} - \tfrac{d_i}{1-\epsilon} \right].
\]
For moderate $\epsilon \geq 0.5$, the variance of $\msehat^{(k)}$ tends to be strongly influenced by the test component (see Figure~\ref{fig:dt-variance}). Conditional on the direct estimate $y_i$, each fold component can be written as
\[
y_i^{(k)} = \frac{y_i}{K} + e_i^{(k)},
\]
where $(e_i^{(1)},\ldots,e_i^{(K)}) \given y_i$ are jointly Gaussian with mean zero and pairwise covariance $\Cov{e_i^{(k)}}{e_i^{(j)} \given y_i} = -d_i/K^2$ for $k \neq j$.
Therefore, $\Cov{(e_i^{(k)})^2}{(e_i^{(j)})^2 \given y_i} = 2d_i^2/K^4 > 0$.
This mechanism suggests that the squared terms in the MSE estimator may induce positive correlation across folds, contributing to $\rho(y) > 0$.

To examine this empirically, we compared repeated single-fold thinning ($R=5$, $\epsilon=0.8$) against multi-fold thinning ($K=5$, yielding the same training fraction $\epsilon = 0.8$). We tested the two approaches across three different equal allocation sampling designs, fitting Fay--Herriot models with varying complexity to the data. Table~\ref{tab:multifold} reports the ratio of variances for the averaged MSE estimator computed across 50 simulated samples under each design.

\begin{table}[ht]
\centering
\begin{tabular}{lcccc}
  \toprule
Design & $p=6$ & $p=18$ & $p=30$ & $p=42$ \\
  \midrule
$n=50$ & 1.23 & 1.26 & 1.29 & 1.30 \\
  $n=75$ & 1.44 & 1.49 & 1.51 & 1.63 \\
  $n=100$ & 1.58 & 1.53 & 1.60 & 1.69 \\
   \bottomrule
\end{tabular}
\caption{Variance ratio (multi-fold / repeated) for the averaged MSE estimator across equal-allocation designs, where $n$ denotes the per-area target sample size. $p$ denotes the number of spatial basis functions included as covariate effects in the Fay--Herriot model, serving as a proxy for model complexity. Ratios above 1 indicate higher variance for multi-fold thinning. Results based on 50 samples per design.}
\label{tab:multifold}
\end{table}

Variance ratios range from 1.23 to 1.69, exceeding 1 in all configurations we have examined. Although these results may differ by dataset and candidate models, this reflects our experience using multi-fold thinning, which seems to under-perform repeated thinning in many other settings. Across all designs, the penalty increases with model complexity, with ratios mostly growing monotonically from $p=6$ to $p=42$ within each row. This model-dependence compounds the challenge identified in Section~\ref{sec:thinning_gap}. Not only does the thinning gap vary across candidate models, but so does the variance reduction from averaging. With repeated thinning, variance reduction scales predictably as $1/R$ regardless of the data or model.

Based on these findings, we recommend repeated single-fold thinning as the default approach. In our experiments, $R \approx 5$ repeats are typically sufficient to stabilize the MSE estimate at modest computational cost (see Section~\ref{sec:eps_repeat}).


\section{Likelihood-Based Validation with Data Thinning}\label{sec:likelihood}

An alternative to MSE-based validation is to evaluate models using predictive log-likelihood. The ideal target is the expected log pointwise predictive density (ELPD) \citep{vehtari_practical_2017} for future observations:
\[
\mathrm{ELPD} = \sum_{i=1}^m \int \log p(\tilde{y}_i \given y)\, f(\tilde{y}_i)\, d\tilde{y}_i,
\]
where $f$ denotes the true data-generating density and $p(\cdot \given y)$ is the model-based predictive density. ELPD measures how well a fitted model predicts genuinely new data.

The challenge is that ELPD cannot be computed directly since we observe only one dataset $y$, not future realizations $\tilde{y}$. Classical information criteria instead approximate out-of-sample predictive performance using an in-sample goodness-of-fit term plus a complexity penalty $\lambda$:
\[
\mathrm{IC} = -2\log p(y \given \hat{\theta}) + \lambda.
\]
AIC \citep{akaike_new_1974} uses $\lambda = 2k$ for models with $k$ parameters, derived from asymptotic arguments under maximum likelihood. However, AIC does not naturally extend to hierarchical or Bayesian settings where the effective complexity is not simply a parameter count. DIC \citep{spiegelhalter_bayesian_2002} adapted this framework for Bayesian models using $\lambda = 2p_D$, where the effective number of parameters $p_D$ is derived from posterior variability of the deviance. WAIC \citep{watanabe_asymptotic_2010} further refined the approach by averaging over the posterior distribution rather than conditioning on a point estimate, providing better theoretical properties for singular models.

Data thinning takes a fundamentally different route. Rather than approximating out-of-sample performance from in-sample quantities, we create genuinely independent train and test sets and evaluate predictive performance directly; no penalty term is required. Under data thinning, the test observation follows $\ytest_i \sim \N{(1-\epsilon)\theta_i}{(1-\epsilon)d_i}$. Given an estimate $\thetatrain_i$ from the training set, we propose the predictive log-likelihood
\[
\ell_{\epsilon} := \sum_{i=1}^m \log \phi\!\left( \ytest_i \agiven (1-\epsilon)\thetatrain_i,\, (1-\epsilon)d_i \right),
\]
where $\phi(\cdot \mid \mu, \sigma^2)$ denotes the Gaussian density. The rescaling $(1-\epsilon)\thetatrain_i$ transforms the training estimate
of $\theta_i$ into a prediction for the test-set mean $(1-\epsilon)\theta_i$. Avoiding complications of how to specify a penalty term, this approach provides an agnostic way to evaluate point-estimates from Bayesian or frequentist models.

The inherent trade-off with this method is that the predictive target differs from the full-data ELPD. Data thinning targets a modified quantity, the thinned-data ELPD:
\[
\mathrm{ELPD}_\epsilon = \sum_{i=1}^m \int \log p(\tilde{y}^{(2)}_i \given \ytrain)\, f_\epsilon(\tilde{y}^{(2)}_i)\, d\tilde{y}^{(2)}_i,
\]
where $\tilde{y}^{(2)}_i$ denotes a hypothetical test observation and $f_\epsilon$ its marginal density under thinning. The relationship between $\mathrm{ELPD}_\epsilon$ and the full-data $\mathrm{ELPD}$ parallels the discussion in Section~\ref{sec:mse_validation}: the trade-off between the thinning gap and estimation variance applies in the likelihood setting as well.

Expanding the Gaussian log-likelihood reveals that maximizing $\ell_{\epsilon}$ is equivalent to minimizing a weighted MSE:
\[
\E[\ytest]{\ell_{\epsilon} \given \ytrain} = C - \frac{1}{2}\sum_{i=1}^m \frac{1-\epsilon}{d_i} \left( \thetatrain_i - \theta_i \right)^2,
\]
where $C$ depends only on known constants that are not model-dependent (see Appendix~\ref{app:lik_mse} for details). The weights $(1-\epsilon)/d_i$ naturally down-weight areas with large sampling variance, which makes the score more stable at the cost of being less sensitive to model performance in precisely the areas where borrowing strength matters most. This represents a potential disadvantage relative to unweighted MSE comparisons, which give equal attention to all areas regardless of direct estimate precision.


\section{Empirical Analysis and Model Comparison}\label{sec:empirical}

We now examine data thinning empirically using the spatial basis selection problem introduced in Section~\ref{sec:motivating_ex}. We first describe the simulation framework shared across the analyses in this section. We then study how the training fraction $\epsilon$ and the number of repeated thinning iterations $R$ affect model selection. Finally, we conduct a full empirical comparison, benchmarking data thinning against existing model selection methods.

\subsection{Simulation Framework}\label{sec:sim_framework}

\paragraph{Data Generation:} The California PUMS data for 2019--2023 comprises approximately $N = 1.76$ million person records across $m = 281$ Public Use Microdata Areas (PUMAs), with employment-to-population rate as the target parameter $\theta_i$ for each area. While the PUMS data is itself a sample of the full population, we treat it as a finite population for the purposes of our simulation and then subsample from it. This provides a finite-population oracle $\theta$ against which all methods can be evaluated.

Our simulation uses two types of sampling designs: equal allocation and proportional-to-population allocation. For each design, we draw $S = 50$ independent samples using stratified Poisson sampling with strata defined by PUMAs.  The within-stratum inclusion probabilities are set proportional to the person weights included with PUMS. Under equal allocation, the expected sample size is held constant across areas at $\E{n_i} \in \{30, 50, 75, 100\}$. The realized sample sizes vary even under equal allocation due to Poisson sampling (e.g., $n_i$ ranges from roughly 28 to 74 when $n = 50$). We refer to these designs by their target $n = \E{n_i}$, and use them to evaluate and illustrate the effect of thinning parameters in Section~\ref{sec:mse_validation} and Section~\ref{sec:eps_repeat}.

We also consider proportional-to-population allocation, where the inclusion probabilities are scaled to achieve overall expected sampling rates of 0.75\%, 1.25\%, and 1.75\% within each PUMA. This design reflects more common sample size variation and are used for comparing different model validation procedures in Section~\ref{sec:method_comparison}.
From each realized sample, we compute Horvitz--Thompson direct estimates $y_i$ using inverse-probability weighting and design-based variance estimates $d_i$ via Taylor linearization.

\paragraph{Candidate Models:} We consider the models described in Section~\ref{sec:motivating_ex}. Model complexity is indexed by the number of spatial basis functions $p \in \{3, 6, 9, \ldots, 60\}$. To construct spatial basis functions, we follow \citet{hughes_dimension_2013}.
Let $A$ denote the $m \times m$ binary adjacency matrix indicating shared borders between areas.  Let $P_X = X(X^\top X)^{-1}X^\top$ project onto the column space of an initial covariate matrix $X$. The Moran operator
\[
G = (I - P_X) A (I - P_X)
\]
captures spatial autocorrelation orthogonal to $X$ \citep{moran_notes_1950, hughes_dimension_2013}.
We take $G_p$ to be the $m \times p$ matrix whose columns are the eigenvectors of $G$ corresponding to the $p$ largest positive eigenvalues.

For the California PUMA geography, $G$ has $114$ positive eigenvalues out of $281$ total, and our candidate grid spans roughly 3--50\% of the available positive spectrum. The \citet{hughes_dimension_2013} heuristic of retaining 10\% of all eigenvectors would suggest $p \approx 28$, which falls near the middle of our grid.

In our application, the initial covariate matrix is simply the intercept, giving the augmented design matrix $[\mathbf{1} \mid G_p]$. The candidate models are Fay--Herriot models with IID random effects:
\[
y_i \ind \N{\theta_i}{d_i}, \quad \theta_i = x_i^\top \beta + u_i, \quad u_i \iid N(0, \sigma^2),
\]
where $x_i^\top$ is the $i$th row of $[\mathbf{1} \mid G_p]$ and $p \in \{3, 6, 9, \ldots, 60\}$. We use a Bayesian implementation and use the posterior mean as our point estimate. For the coefficients $\beta$, we place an improper flat prior $\pi(\beta)\propto 1$. For the random-effect variance, we use a proper inverse-gamma prior $\sigma^2 \sim \mathrm{IG}(a,b)$ with $a=b=0.001$. The chosen parameters provide a diffuse prior over practically relevant values.

\paragraph{Evaluation:} Here, model selection is the task of
selecting the basis count $p$ within a single nested family of
Fay--Herriot models. We evaluate the selected number of basis functions
against the \textit{average oracle basis} $p^*$, the number of basis
functions that minimizes mean squared error averaged across all $S$
simulated samples:
\[
p^* = \arg\min_p \frac{1}{S} \sum_{s=1}^{S} \sum_{i=1}^m
  \bigl(\tilde{\theta}_i^{(s)}(p) - \theta_i\bigr)^2,
\]
where $\tilde{\theta}_i^{(s)}(p)$ denotes the posterior mean for area
$i$ from the model with $p$ basis functions fitted to the $s$th sample.
The target $\theta_i$ is the fixed finite-population mean from
Assumption~\ref{assump:framework}, common to all candidate models; only
the predictor $\tilde{\theta}_i^{(s)}(p)$ varies with $p$.
Consequently, $p^*$ is not model-dependent: it is a single reference
value, determined by the data-generating truth and the survey design.
The average oracle basis is remarkably stable: $p^* = 15$ for all three
proportional allocation designs and also for equal allocation designs
except $n = 30$, where it drops to $p^* = 12$. Both values are well
below the $p \approx 28$ suggested by the \citet{hughes_dimension_2013}
heuristic.

Let $\check{p}_s$ denote the basis count selected by a given method on
sample $s = 1, \ldots, S$. We report two complementary metrics:
\begin{itemize}
  \item \textit{Root mean squared error}:
    $\mathrm{RMSE} := \bigl(S^{-1} \sum_s
    (\check{p}_s - p^*)^2\bigr)^{1/2}$, measuring typical error in the
    selected basis count.
  \item \textit{Mean bias}:
    $\mathrm{Mean\ Bias} := S^{-1} \sum_s (\check{p}_s - p^*)$, indicating
    whether a method systematically under-selects (negative) or
    over-selects (positive) relative to the oracle.
\end{itemize}
RMSE and bias together characterize a method's selection accuracy and
directional tendency.


\subsection{Effect of \texorpdfstring{$\epsilon$}{epsilon} and Repeats \texorpdfstring{$R$}{R}}\label{sec:eps_repeat}

We first examine the effect of $\epsilon$ and repeats $R$ on model
selection under the equal allocation design.
Figure~\ref{fig:basis_selection} shows the impact of the thinning
parameter $\epsilon$ and repeats $R$ on model selection from the
Fay--Herriot models using spatial basis fixed effects. The most striking
pattern in Figure~\ref{fig:basis_selection}(b) is the systematic
under-selection at small training fractions $\epsilon$. For
$\epsilon < 0.5$, the mean bias is negative across all values of $R$ and
all sample sizes, with under-selection of 7--12 basis functions on average. For small
$\epsilon$, increasing the number of repeated thinnings $R$ seems to
simply sharpen this bias. As $\epsilon$ increases, the bias decreases
and crosses zero near $\epsilon = 0.6$--$0.8$ for the larger sample
sizes.

This pattern confirms the theoretical analysis in
Section~\ref{sec:thinning_gap}. The thinning gap is larger for complex
models (Proposition~\ref{prop:thinning_gap2}) and is amplified at small
$\epsilon$ (Corollary~\ref{cor:thinning_gap_monotone}). Because the gap
penalizes complex models more heavily, validation systematically favors
models that are simpler than the oracle when the training fraction is
insufficient, leading to more conservative model choices.

The RMSE in the upper panel~(a) tells a complementary story. Although
high $\epsilon$ values reduce the oracle gap, they introduce substantial
variability that is detrimental for model selection. For $R = 1$, RMSE
increases sharply beyond $\epsilon \approx 0.5$. This high-$\epsilon$
instability reflects the variance properties highlighted in
Section~\ref{sec:variance}. As $\epsilon$ approaches 1, validation based
on $\ytest$ becomes increasingly unreliable. Repeated thinning mitigates
this effect but cannot eliminate it entirely.

Overall, we see here that the choice of $\epsilon$ is much more
important than the number of repeats $R$. The tension between the
thinning gap and variance creates a favorable range at moderate
$\epsilon$. Across all sample sizes, RMSE is minimized in the range
$\epsilon \approx 0.5$--$0.7$. Within this range, increasing $R$
consistently improves performance by reducing the variability inherent to
single splits, though the improvement from $R=3$ to $R=5$ seems to
indicate diminishing returns.

\begin{figure}[htbp]
    \centering
    \includegraphics[width=0.95\linewidth]{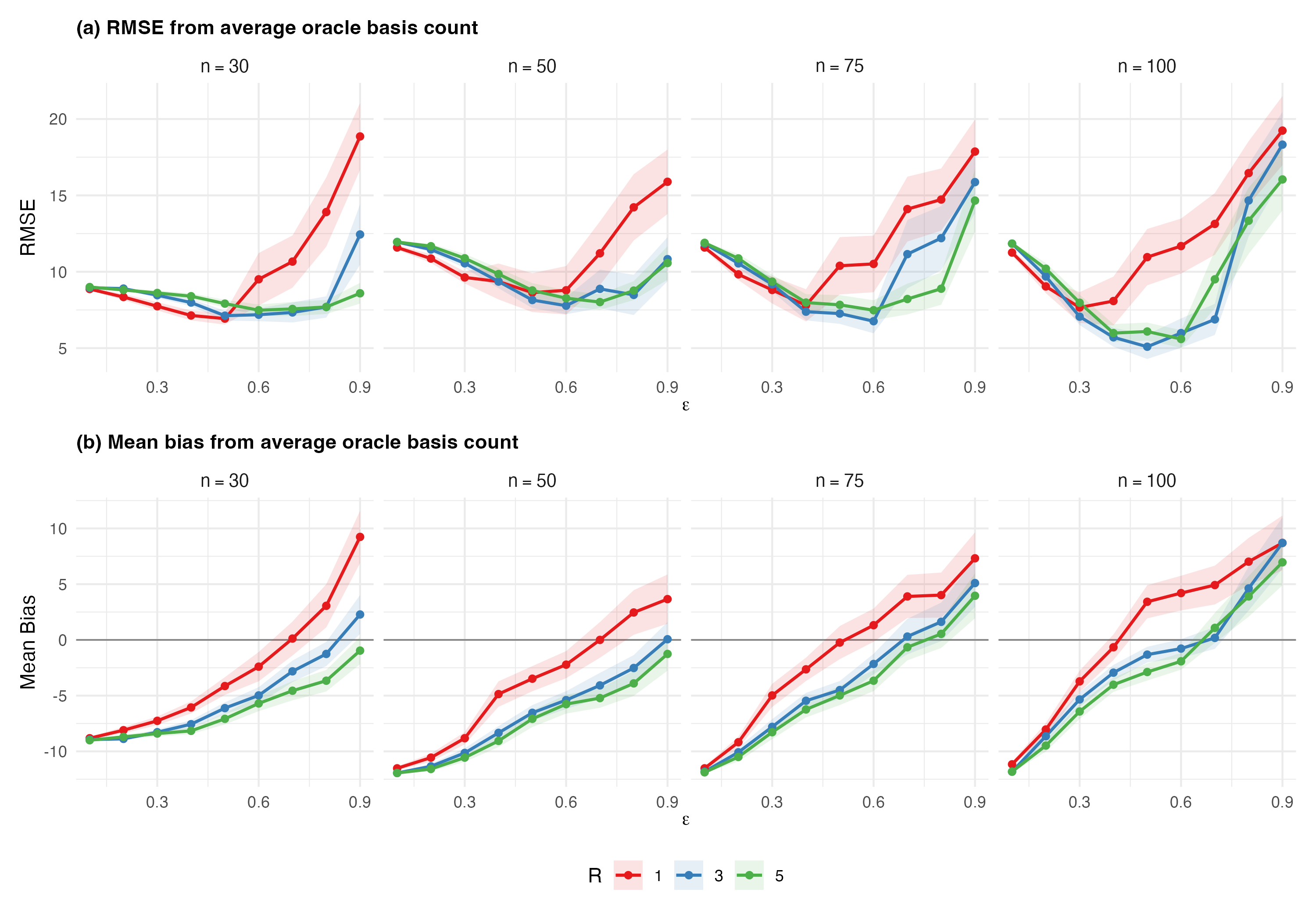}
    \caption{Effect of the training fraction $\epsilon$ and the number of repeats $R \in \{1, 3, 5\}$ on basis selection under equal-allocation designs with target sample sizes $n$. Shaded ribbons indicate $\pm 1$ standard errors of the mean, taken over 50 simulated datasets. Panel~(a): RMSE from the average oracle basis count. Panel~(b): Mean bias; negative values indicate under-selection.}
    \label{fig:basis_selection}
\end{figure}

\subsection{Comparison with Existing Methods}\label{sec:method_comparison}

We now benchmark data thinning against established model selection
approaches using the proportional-to-population allocation designs and
the candidate models described in Section~\ref{sec:sim_framework}.

\paragraph{Methods Compared:} The two data thinning approaches are the
MSE estimator (DT-MSE) from Section~\ref{sec:mse_validation} and the
negative log-likelihood score (DT-NLL) from
Section~\ref{sec:likelihood}, both with $\epsilon = 0.6$ and $R = 5$
repeats based on the analysis in Section~\ref{sec:eps_repeat}.

We compare against three established approaches. DIC
\citep{spiegelhalter_bayesian_2002} and WAIC
\citep{watanabe_asymptotic_2010} are Bayesian information criteria that
balance in-sample fit against a complexity penalty; both are discussed
further in Section~\ref{sec:likelihood}. Neither produces genuinely
out-of-sample evaluations.

The ESIM (Empirical Simulation) approach is commonly used in SAE
\citep{bradley_multivariate_2015, janicki_bayesian_2022}. The approach
generates $\ell = 1, \ldots, L$ synthetic direct estimates
$z^{(\ell)}_i := y_i + e^{(\ell)}_i$ where
$e^{(\ell)}_i \ind N(0, d_i)$, fits candidate models to each
$z^{(\ell)}$, and validates against the original direct estimates by
setting $\theta_i := y_i$. We use $L = 100$ iterations per sample. ESIM
requires only area-level summary statistics and is estimator-agnostic,
but its core assumption---that direct estimates equal true area
means---is difficult to justify in precisely the small-area settings
where model-based estimation is most needed.

\begin{table}[ht]
\centering
\begin{tabular}{l rrrr}
  \toprule
  Method & 0.75\% & 1.25\% & 1.75\% & Overall \\
  \midrule
  DT-MSE $\epsilon$=0.6 & 9.20 & 5.91 & 5.89 & 7.17 \\
  DT-NLL $\epsilon$=0.6 & 8.89 & 5.68 & 6.61 & 7.19 \\
  DIC & 6.98 & 6.79 & 10.60 & 8.31 \\
  WAIC & 6.77 & 9.49 & 14.26 & 10.63 \\
  ESIM & 10.65 & 6.77 & 3.72 & 7.60 \\
  \bottomrule
\end{tabular}
\caption{RMSE of the selected basis count from the average oracle
($p^* = 15$) by method and design across $S = 50$ simulation replicates.
Columns indicate proportional allocation designs with overall sampling
rates of 0.75\%, 1.25\%, and 1.75\%. Overall RMSE is computed by pooling all $150$ simulated datasets across the three designs.}
\label{tab:comparison}
\end{table}

Table~\ref{tab:comparison} reports RMSE from the average oracle
basis count, and Figure~\ref{fig:method_comparison} shows the
distribution of selected basis counts across $S = 50$ simulated samples
per design. DT-MSE and DT-NLL achieve the lowest overall RMSE (7.17 and
7.19), followed by ESIM (7.60). The per-design results reveal a clear pattern: DIC and WAIC perform best at low precision (0.75\% PA), ESIM dominates at high precision (1.75\% PA), and data thinning leads at 1.25\% PA while remaining competitive throughout.

The per-design results show how each method performs across the range of
sampling precision. At 0.75\% PA, where sampling noise is highest, DIC
and WAIC perform best while ESIM struggles (RMSE 10.65). At 1.25\% PA,
the data thinning methods lead: DT-NLL achieves the lowest RMSE of any
method at any design (5.68), followed closely by DT-MSE (5.91). At
1.75\% PA, ESIM dominates with RMSE of 3.72, while DIC and WAIC
deteriorate sharply (10.60 and 14.26). Notably, DT-MSE maintains nearly
the same accuracy at 1.75\% as at 1.25\% (5.89 vs.\ 5.91), while other
methods show large swings.

The patterns in Figure~\ref{fig:method_comparison} and
Table~\ref{tab:comparison} are striking. At 0.75\% PA, all methods
under-select relative to the oracle, reflecting the difficulty of model
selection when sampling noise is high. As sample size increases, the
methods diverge sharply. DIC and WAIC spread upward with outliers
reaching over 40 basis functions above the oracle at 1.75\% PA. Their
mean bias swings from under-selection to substantial over-selection
across designs. This
reflects a structural limitation of information criteria in hierarchical
models: the log-likelihood improvement from finer spatial structure
scales with $1/d_i$, while penalty terms remain bounded by the number of
areas $m$ \citep{gelman_understanding_2014}.

ESIM remains tightly concentrated below the oracle across all designs
($-10.1$ to $-1.4$), but its accuracy improves substantially as
precision increases. ESIM performs best when the direct estimates $y_i$
are close to the true parameters $\theta_i$, since its validation target
$\theta_i := y_i$ is then approximately correct. As noted in
Appendix~\ref{app:esim}, ESIM is equivalent to data fission with a misspecified
target: the correct validation target under data fission would be
$(\theta_i + y_i)/2$, not $y_i$. This mismatch is small when
$y_i \approx \theta_i$ but grows with sampling noise.

The data thinning methods sit between these extremes, with distributions
that track the oracle most closely at 1.25\% and 1.75\% PA, though with
a slight downward tendency reflecting the thinning gap inherent in
training on partial data. DT-MSE has a moderate bias swing ($-7.0$ to
$+0.3$), landing near zero at the most precise design. DT-NLL shows
notable spread at 1.75\% PA and drifts further positive ($+2.0$). As a
likelihood-based method, DT-NLL is subject to the same mechanism that
drives DIC and WAIC to over-select, but with greater protection from
out-of-sample evaluation. At 1.75\% PA, the DT and ESIM distributions
look notably similar, consistent with the connection between data
thinning and ESIM developed in Appendix~\ref{app:esim}.

\begin{figure}[ht]
    \centering
    \includegraphics[width=\textwidth]{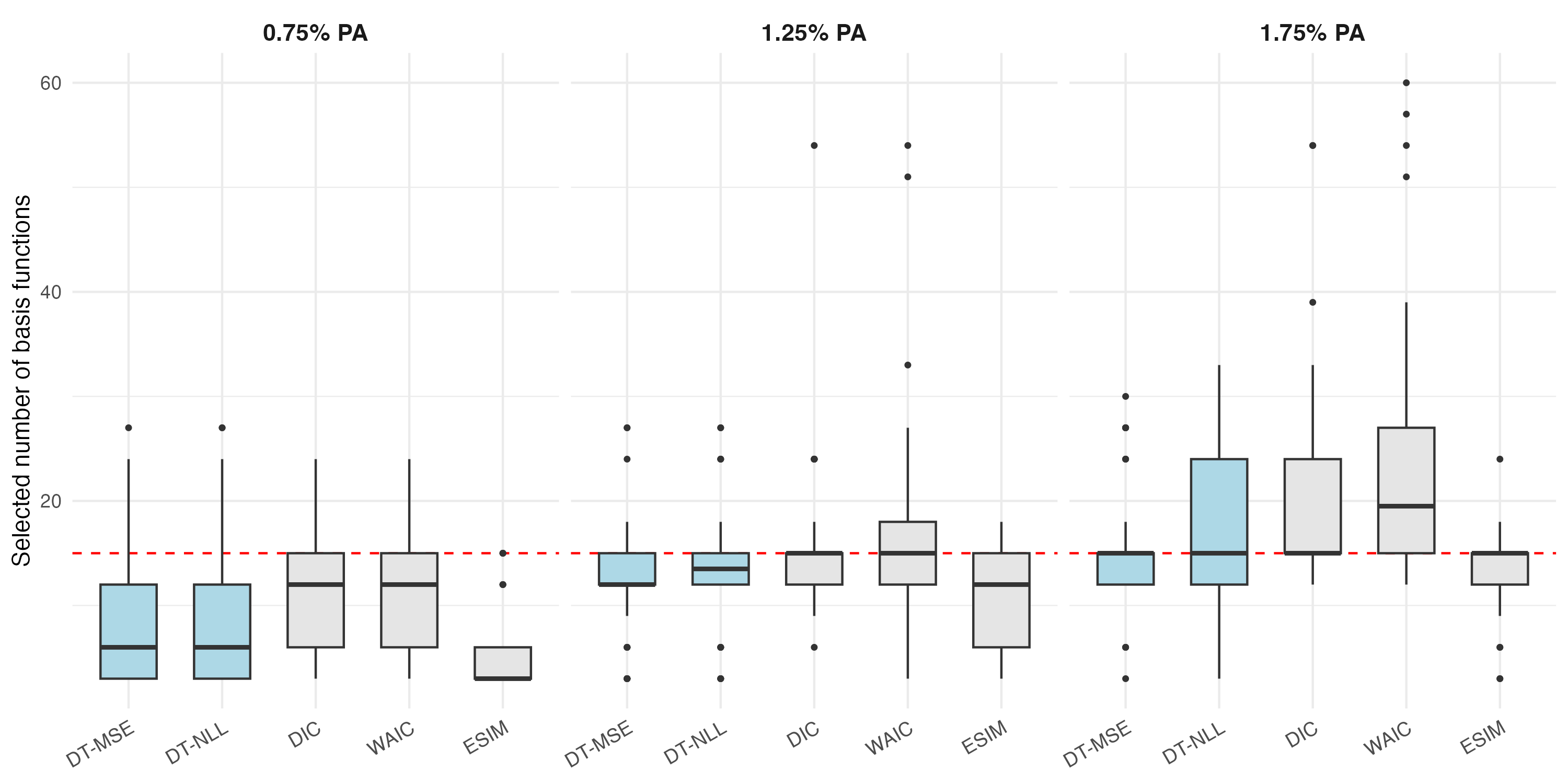}
    \caption{Distribution of selected basis function counts across
    methods and proportional allocation (PA) designs ($S = 50$ simulated
    samples per design). The dashed red line marks the average oracle
    ($p^* = 15$). Data thinning methods are tinted in light blue.}
    \label{fig:method_comparison}
\end{figure}

When sampling noise is high, information criteria are effective and
computationally cheap. Any approach that introduces additional noise,
whether through data splitting (data thinning) or synthetic perturbation
(ESIM), is likely to be detrimental in small-sample settings. However,
data thinning's advantage is reliability across the full range of
designs. Where DIC and WAIC over-select at high precision and ESIM
collapses at low precision, data thinning remains competitive
throughout. With $R = 5$ repeats, data thinning requires fitting each
candidate model only five times, compared to ESIM's $L = 100$
iterations. Unlike DIC and WAIC, which are restricted to Bayesian models, both DT and ESIM are estimator-agnostic, but data thinning requires no assumptions beyond Assumption~\ref{assump:framework}.

\section{Discussion and Future Work}\label{sec:discussion}

This paper investigated data thinning as a model validation tool for SAE. Our theoretical analysis reveals a fundamental trade-off: the thinning gap between thinned-data and full-data performance metrics decreases with the training fraction $\epsilon$, while the variance of the MSE estimator increases. This trade-off is model-dependent, with complex models incurring larger thinning gaps, and no single $\epsilon$ is uniformly optimal across candidate models.

Nevertheless, data thinning offers a unified, out-of-sample validation framework that has been sorely missing in SAE. It relies only on Gaussianity of the direct estimates and known sampling variances (Assumption~\ref{assump:framework}), both standard in area-level modeling. Unlike information criteria, it requires no penalty approximation; unlike ESIM, it makes no assumption that direct estimates equal true area means. Based on empirical analysis, we recommend $\epsilon \approx 0.5$--$0.7$ with repeated thinning $R\approx 5$, which approximately equalizes estimation errors across models while keeping variance under control. Our design-based simulations show that data thinning with the recommended settings provides competitive and consistent performance across sampling designs, avoiding the failure modes exhibited by existing methods. Although our example focused on a model selection task involving relatively simple Fay--Herriot models with spatial basis functions as covariates, data thinning can be applied to compare any area-level models satisfying Assumption~\ref{assump:framework}, including models with different random effect structures.

The theoretical framework we develop reveals properties of data thinning that we believe extend beyond SAE. The same trade-off discussed in this paper may arise in other settings where thinning is used to validate models with different complexity or shrinkage behavior on a single dataset.  More broadly, data thinning may offer a useful theoretical lens for studying model validation. Unlike cross-validation, which operates on discrete folds, data thinning provides a continuous parameter $\epsilon \in (0,1)$ governing the train-test allocation. For the family of distributions and sufficient statistics that can be thinned, the components have known, tractable distributions \citep{dharamshi_generalized_2025}. This structure enabled the closed-form thinning gap analysis in this paper and may facilitate sharper theoretical results about single-dataset validation than are available for sample-splitting approaches.

The connection between data thinning and cross-validation is worth highlighting. For example, data thinning makes the difficulty of estimating in-sample predictive error, pointed out in \citet{bates_cross-validation_2024}, extremely obvious; conditioning on the full data, the training and test sets under data thinning are perfectly negatively correlated. Motivated by data thinning, \citet{liu_cross-validation_2026} proposed a Gaussian randomization scheme for constructing train-test pairs that achieve lower variance in error estimation than standard cross-validation.

In SAE, the connection between data thinning and cross-validation arises as well. Area-level cross-validation can be viewed as the limiting case of data thinning where held-out areas receive $\epsilon_i = 0$. Our finding that optimal $\epsilon$ is model-dependent aligns with recent work by \citet{mcalinn_determining_2025}, who show that the optimal number of folds $K$ in cross-validation depends on both data and model. In SAE, this problem is compounded by heterogeneity in sampling variances $d_i$. Determining an optimal fold assignment that accounts for this heterogeneity is a challenging combinatorial problem that data thinning sidesteps entirely.

Several limitations of this work suggest directions for future research. Our theoretical results treat sampling variances $d_i$ as known. While this is standard in area-level SAE, the $d_i$ used in practice are themselves design-based estimates that introduce uncertainty not accounted for in our framework. Work by \citet{dharamshi_decomposing_2025} on data thinning with estimated variances provides relevant theoretical groundwork, though their focus is on models that explicitly estimate variance components rather than design-based variance estimation.  Our recommended $\epsilon$ is also uniform across areas, but the area-specific variance results in Proposition~\ref{prop:fh_var_minimize} suggest that optimizing $\epsilon_i$ by area could improve performance. More generally, our analysis focuses on area-level Fay--Herriot models with Gaussian likelihoods; unit-level models and non-Gaussian extensions for count data remain unexplored. Whether analogous thinning gap phenomena arise in other data thinning applications is an open question worth investigating. Constructing multi-$\epsilon$ estimators that reduce the thinning gap without inflating variance remains a direction for future work.

\section*{Acknowledgments and Disclosure Statement}

The authors used Anthropic Claude Sonnet and Opus (v4.5--4.6) to assist with coding, review of mathematical proofs, LaTeX formatting, and manuscript editing.
The authors report there are no competing interests to declare.

\section*{Reproducibility}

The code needed to reproduce the results in this paper is available at \url{https://github.com/sho-kawano/dt_basis_select} (archived at \url{https://doi.org/10.5281/zenodo.20724385}; \citealp{kawano_dt_code_2026}) and is written in R v4.5.1 \citep{r_core_team_r_2025}. All data are from
the 2019--2023 American Community Survey PUMS, publicly available from the U.S.\ Census Bureau.

Direct estimates and design-based variances are computed using the \texttt{survey} package \citep{lumley_survey_2024}. The adjacency
matrix for spatial basis construction is obtained from PUMA shapefiles via the \texttt{tigris} package \citep{walker_tigris_2023}. Fay--Herriot models are fit using a Gibbs sampler implemented via custom R code.

\bibliographystyle{apalike}
\bibliography{all_ref}


\newpage
\section{Appendix}
\subsection{Connection between Empirical Simulation and Data Fission/Thinning}
\label{app:esim}

A simulation strategy frequently used in the small area literature, which we refer to as empirical simulation (ESIM), relies solely on area-level summary statistics rather than microdata \citep{bradley_multivariate_2015, janicki_bayesian_2022}. The procedure generates synthetic direct estimates by injecting additional noise into the observed data. This mechanism shares a structural similarity with data fission \citep{leiner_data_2025}, a close relative to data thinning.

However, the two methods diverge in their validation logic. ESIM generates a perturbed training set $\ytrain_i$ but treats the original noisy estimate $y_i$ as the fixed ground truth $\theta_i$.

\begin{algorithm}[H]
\caption{Empirical Simulation Study (ESIM)}
\label{alg:esim}
\begin{algorithmic}[1]
\Require Direct estimate $y_i \sim \N{\theta_i}{d_i}$ with known variance $d_i$
\State \textbf{Assumption:} Treat observed $y_i$ as the true mean, setting $\theta_i := y_i$
\State Draw training observation $\ytrain_i \mid \theta_i \sim \N{\theta_i}{d_i}$
\State Set validation target $\ytest_i := y_i$
\State \Return Training observation $\ytrain_i$ and target $\theta_i$
\end{algorithmic}
\end{algorithm}

The plug-in step $\theta_i := y_i$ ignores the sampling error inherent in $y_i$. Data fission avoids this plug-in assumption. It generates training data via a similar noise injection (randomizing the data) but accounts for the noise in the validation step. Rather than validating against a fixed point, it validates against the conditional distribution of the remaining data given the randomized training component.

\begin{algorithm}[H]
\caption{Data Fission (Gaussian case, $\tau=1$)}
\label{alg:dependent_df}
\begin{algorithmic}[1]
\Require Direct estimate $y_i \sim \N{\theta_i}{d_i}$ with known variance $d_i$
\State Draw auxiliary noise $e_i \sim \N{0}{d_i}$
\State Construct training observation $\ytrain_i := y_i + e_i$
\State Retain original data for testing $\ytest_i := y_i$
\State \textbf{Inference:} Validate based on the conditional law:
\[
\ytest_i \mid \ytrain_i \sim \N{\frac{\ytrain_i + \theta_i}{2}}{\frac{d_i}{2}}
\]
\end{algorithmic}
\end{algorithm}

In data fission, $\ytrain_i$ and $\ytest_i$ are correlated. Data thinning simplifies this framework by transforming the components into marginally independent variables. \citet{neufeld_data_2024} show that in the Gaussian family, data fission is equivalent to data thinning up to a rescaling. Specifically, the independent splits $\ytrain_i, \ytest_i$ used in Algorithm~\ref{alg:thin} carry the same information as the correlated components in data fission, but the independence property allows for the simpler, intuitive validation procedures described in Section~\ref{sec:mse_validation}.

\subsection{Data Thinning: Two Direct Estimates from One Sample}\label{app:thinning_maps}

\begin{figure}[H]
    \centering
    \includegraphics[width=\textwidth]{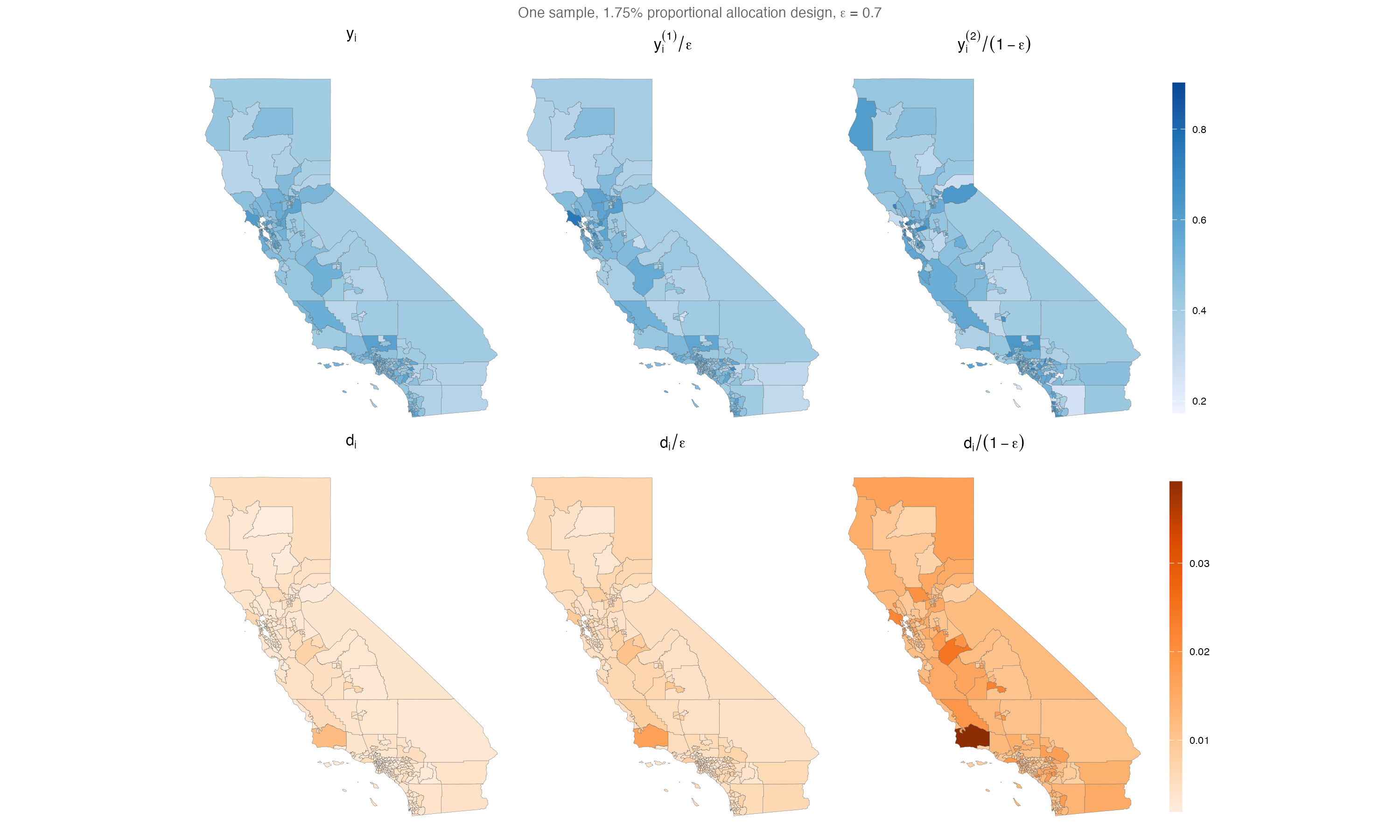}
    \caption{Visualization of how data thinning splits the original direct estimate into two. The sample is drawn using a 1.75\% proportional allocation design with $\epsilon = 0.7$ using the California PUMS data from Section~\ref{sec:motivating_ex}. Top row: the original direct estimates $y_i$ (left), the scaled training data $\ytrain_i/\epsilon$ (center), and the scaled test data $\ytest_i/(1-\epsilon)$ (right). Bottom row: the corresponding sampling variances $d_i$, $d_i/\epsilon$, and $d_i/(1-\epsilon)$.  The test-set variance inflation is visually apparent as it only has $0.3$ of the effective sample size compared to the original.}
    \label{fig:thinning_maps}
\end{figure}

\subsection{Proof of Proposition~\ref{prop:thinning_gap2}: Thinning Gap Under Estimated $\beta$}
\label{app:thinning_gap2}

The proof builds on the MSE decomposition of \citet{prasad_estimation_1990}, applying it separately to the full-data BLUP and thinned-data BLUP.

\begin{proof}
Under the Fay--Herriot model with known $\sigma^2$, the BLUP for area $i$ is
\[
\tilde{\theta}_i = \gamma_i \, y_i + (1-\gamma_i) \, x_i^\top \tilde{\beta},
\]
where $\tilde{\beta} = \left(\sum_{j=1}^m \frac{x_j x_j^\top}{\sigma^2 + d_j}\right)^{-1} \sum_{j=1}^m \frac{x_j }{\sigma^2 + d_j} y_j$ is the weighted least-squares estimator.

By the Prasad--Rao decomposition \citep{prasad_estimation_1990}, the MSE decomposes as
\[
\E{(\tilde{\theta}_i - \theta_i)^2} = g_{1i} + g_{2i},
\]
where $g_{1i} = \gamma_i d_i$ is the prediction variance (identical to the known-$\beta$ case) and
\[
g_{2i} = (1-\gamma_i)^2 \, x_i^\top \left[\sum_{j=1}^m \frac{x_j x_j^\top}{\sigma^2 + d_j}\right]^{-1} x_i
\]
captures the contribution from estimating $\beta$.

For the thinned estimator with effective sampling variance $d_i/\epsilon$, the same decomposition gives
\[
\E{(\tilde{\theta}_i^{(1)} - \theta_i)^2} = g_{1i}(\epsilon) + g_{2i}(\epsilon),
\]
where $g_{1i}(\epsilon) = \gamma_i(\epsilon) \cdot d_i/\epsilon$ and
\[
g_{2i}(\epsilon) = (1-\gamma_i(\epsilon))^2 \, x_i^\top \left[\sum_{j=1}^m \frac{x_j x_j^\top}{\sigma^2 + d_j/\epsilon}\right]^{-1} x_i.
\]
The per-area thinning gap is therefore
\[
\Delta_i(\epsilon) = \big[g_{1i}(\epsilon) - g_{1i}\big] + \big[g_{2i}(\epsilon) - g_{2i}\big].
\]
By Proposition~\ref{prop:thinning_gap}, the first bracket is positive. For the second bracket, note that $\gamma_i(\epsilon) < \gamma_i$ implies $(1-\gamma_i(\epsilon))^2 > (1-\gamma_i)^2$. Additionally, since $d_j/\epsilon > d_j$, the weights in the precision matrix $\sum_j (\sigma^2 + d_j/\epsilon)^{-1} x_j x_j^\top$ are smaller for the thinned data, making its inverse larger. Both factors are larger, so $g_{2i}(\epsilon) > g_{2i}$.

\paragraph{Intercept-only simplification.}
When $x_i = 1$ for all $i$, the matrix inverse reduces to a scalar:
\[
g_{2i} = (1-\gamma_i)^2 \cdot \frac{1}{w}, \qquad g_{2i}(\epsilon) = (1-\gamma_i(\epsilon))^2 \cdot \frac{1}{w(\epsilon)},
\]
where $w = \sum_{j=1}^m (\sigma^2 + d_j)^{-1}$ and $w(\epsilon) = \sum_{j=1}^m (\sigma^2 + d_j/\epsilon)^{-1}$.
\end{proof}


\subsection{Proof of Corollary~\ref{cor:thinning_gap_monotone}: Monotonicity of the Thinning Gap}
\label{app:thinning_gap_monotone}

In Proposition~\ref{prop:thinning_gap2}, we established that the thinning gap is
\[\msethin - \msefull = \frac{1}{m} \sum_{i=1}^m \Delta_i(\epsilon),\]
where
\[\Delta_i(\epsilon) = \underbrace{\frac{1-\epsilon}{\epsilon} \cdot \gamma_i(\epsilon) \, \gamma_i}_{\text{systematic}}\, d_i + \underbrace{\big[g_{2i}(\epsilon) - g_{2i}\big].}_{\text{parameter uncertainty}}\]
The proof proceeds by differentiating each term with respect to $\epsilon$. It suffices to show that both the (i) systematic and (ii) the parameter uncertainty gaps are strictly decreasing.

\paragraph{(i) Systematic gap}

\begin{proof}
   Given the proof of Proposition~\ref{prop:thinning_gap} we can rewrite the systematic gap as $\gamma_i(\epsilon) d_i / \epsilon - \gamma_i d_i$. The first $\epsilon$ dependent term can be re-written as
\begin{align*}
\frac{\gamma_i(\epsilon) d_i }{\epsilon} &= \frac{1}{\epsilon} \cdot \frac{\sigma^2 d_i}{\sigma^2  + d_i / \epsilon } = \frac{1}{\epsilon} \cdot \frac{ \epsilon \sigma^2  d_i}{\epsilon \sigma^2  + d_i } = \frac{\sigma^2 d_i}{\epsilon \sigma^2  + d_i }.
\end{align*}
Thus we can evaluate the derivative of the systematic gap:
\begin{align*}
\frac{d}{d\epsilon}\Big[\frac{1-\epsilon}{\epsilon} \gamma_i(\epsilon) \gamma_i d_i\Big]
&= \frac{d}{d\epsilon}\Big[\frac{\gamma_i(\epsilon) d_i}{\epsilon} - \gamma_i d_i\Big]
 = \frac{d}{d\epsilon}\Big[\frac{\gamma_i(\epsilon) d_i}{\epsilon}\Big] \\
&= \frac{d}{d\epsilon}\Big[\frac{\sigma^2 d_i}{\epsilon\sigma^2 + d_i}\Big]
 = \frac{-\sigma^4 d_i}{(\epsilon\sigma^2 + d_i)^2} < 0.
\end{align*}
Since this is negative for all $\epsilon$, the systematic term is strictly decreasing.

\paragraph{(ii) Parameter uncertainty gap}

Define $D(\epsilon) = \mathrm{diag}\big((\sigma^2 + d_j/\epsilon)^{-1}\big)$.
We show that
$g_{2i}(\epsilon) = (1-\gamma_i(\epsilon))^2 \cdot x_i^\top (X^\top D(\epsilon) X)^{-1}x_i$
is strictly decreasing in $\epsilon$.

\textit{Step 1: Derivative of the quadratic form.}
The $j$-th diagonal entry of $D(\epsilon)$ can be written as
\[
[D(\epsilon)]_{jj} = \frac{\epsilon}{\epsilon\sigma^2 + d_j}.
\]
Its derivative is
\[
[D(\epsilon)]_{jj}' = \frac{d_j}{(\epsilon\sigma^2 + d_j)^2} > 0,
\]
so $D'(\epsilon)$ is a diagonal matrix with strictly positive
entries. Let $M(\epsilon) = X^\top D(\epsilon) X$, so that
$M'(\epsilon) = X^\top D'(\epsilon) X$. Since $D'(\epsilon)$ is
positive definite and $X$ has full column rank, $M'(\epsilon)$ is
positive definite. Using the matrix identity
$(M^{-1})' = -M^{-1}M'M^{-1}$,
\[
\frac{d}{d\epsilon}\big(x_i^\top M(\epsilon)^{-1}x_i\big)
= -x_i^\top M(\epsilon)^{-1} M'(\epsilon)\, M(\epsilon)^{-1}x_i.
\]
Since $M'(\epsilon)$ is positive definite and $M(\epsilon)^{-1}$ is
nonsingular, for any $x_i \neq 0$ we have
$M(\epsilon)^{-1}x_i \neq 0$, and so
$x_i^\top M(\epsilon)^{-1} M'(\epsilon)\, M(\epsilon)^{-1}x_i > 0$. The derivative is therefore strictly negative.

\textit{Step 2: Derivative of the squared shrinkage.}
Writing $(1-\gamma_i(\epsilon))^2 = d_i^2/(\epsilon\sigma^2 + d_i)^2$,
we have
\[
\frac{d}{d\epsilon}(1-\gamma_i(\epsilon))^2
= -\frac{2\sigma^2 d_i^2}{(\epsilon\sigma^2 + d_i)^3} < 0.
\]

\textit{Conclusion:} Let $a(\epsilon) = (1-\gamma_i(\epsilon))^2$ and
$b(\epsilon) = x_i^\top M(\epsilon)^{-1}x_i$. Both are strictly
positive with strictly negative derivatives. By the product rule,
$(ab)' = a'b + ab' < 0$, so $g_{2i}(\epsilon) = a(\epsilon)\,
b(\epsilon)$ is strictly decreasing in $\epsilon$.
\end{proof}

\subsection{Proposition~\ref{prop:g2i_covariates} and Proof}
\label{app:g2i_covariates}

\begin{proposition}
\label{prop:g2i_covariates}
(Monotonicity of $g_{2i}$ in model dimension) For nested Fay--Herriot models with fixed $\sigma^2$, the $g_{2i}$ terms are non-decreasing in the number of covariates $p$.
\end{proposition}

\begin{proof}
Let $k \in \{1,2\}$ index two nested models $\mathcal{M}_1$ and
$\mathcal{M}_2$ with covariate matrices $X_1$ and $X_2$ satisfying
$\mathrm{col}(X_1) \subseteq \mathrm{col}(X_2)$ and $\sigma^2$ is fixed. Since $\sigma^2$ is
fixed, the shrinkage factors $\gamma_i$ and the weight matrix
$D = D(1) = \mathrm{diag}\big((\sigma^2 + d_j)^{-1}\big)$ are common to
both models. It suffices to show that
$x_{ki}^\top (X_k^\top D X_k)^{-1} x_{ki}$ is at least as large for
$k=2$ as for $k=1$.

Define $Q_k = X_k(X_k^\top D X_k)^{-1}X_k^\top$. Then
$P_k = D^{1/2}\, Q_k\, D^{1/2}$ is the orthogonal projection onto
$\mathrm{col}(D^{1/2} X_k)$. Since $D^{1/2}$ is nonsingular,
$\mathrm{col}(X_1) \subseteq \mathrm{col}(X_2)$ implies
$\mathrm{col}(D^{1/2}X_1) \subseteq \mathrm{col}(D^{1/2}X_2)$, so
$P_2 - P_1$ is itself an orthogonal projection (onto the complement
of $\mathrm{col}(D^{1/2}X_1)$ within $\mathrm{col}(D^{1/2}X_2)$) and
hence positive semidefinite:
\[
D^{1/2}(Q_2 - Q_1)\,D^{1/2} \succeq 0.
\]
Since $D^{1/2}$ is nonsingular, congruence gives
$Q_2 - Q_1 \succeq 0$. Evaluating the $i$-th diagonal entry:
\[
x_{2i}^\top (X_2^\top D X_2)^{-1} x_{2i} \;\geq\;
x_{1i}^\top (X_1^\top D X_1)^{-1} x_{1i},
\]
which, together with constant $(1-\gamma_i)^2$, gives the result.
\end{proof}


\subsection{Derivation: Variance of the MSE estimator} \label{app:var_deriv}
\begin{lemma}[Moments of a squared Gaussian]\label{lem:sqnorm}
If $X \sim N(0,\sigma^2)$, then
\[
  \E{X^2} = \sigma^2
  \qquad \text{and} \qquad
  \Var{X^2} = 2\sigma^4.
\]
\end{lemma}

Using this lemma we derive the variance of the MSE estimator first by conditioning on the training set $\ytrain$.
\begin{proof}
We have the MSE estimator
\[
\msehat
:= \frac{1}{m}\sum_{i=1}^m
\left[
\left( \thetatrain_i - \tfrac{1}{1-\epsilon}\ytest_i \right)^{\!2}
- \tfrac{d_i}{1-\epsilon}
\right]
\]

For ease of notation, define
\[
  \delta_i := \thetatrain_i - \theta_i,
  \qquad
  \eta_i := \tfrac{1}{1-\epsilon}\,\ytest_i - \theta_i.
\]
Note that $\delta_i \indep \eta_i$ (since $\thetatrain_i$ depends only on $\ytrain$) and $\eta_i \sim \N{0}{\tfrac{d_i}{1-\epsilon}}$.

Each term inside the sum is
\[
(\delta_i - \eta_i)^2 - \tfrac{d_i}{1-\epsilon}.
\]

By the law of total variance,
\begin{align*}
\Varthin{\msehat}
&= \E[\ytrain]{ \Var[\ytest]{ \msehat \agiven \ytrain} }
 + \Var[\ytrain]{ \E[\ytest]{\msehat \agiven \ytrain } }.
\end{align*}

We can then plug in the term inside the sum to get
\begin{align*}
\Varthin{\msehat}
&= \E[\ytrain]{ \Var[\ytest]{ \tfrac{1}{m}\sum_{i=1}^m \big[(\delta_i-\eta_i)^2 - \tfrac{d_i}{1-\epsilon}\big] \agiven \ytrain } }\\
&\quad + \Var[\ytrain]{ \E[\ytest]{ \tfrac{1}{m}\sum_{i=1}^m \big[(\delta_i-\eta_i)^2 - \tfrac{d_i}{1-\epsilon}\big] \agiven \ytrain } }.
\end{align*}

\paragraph{First term.}
Given $\ytrain$, the $\delta_i$ are constants and $\{\eta_i\}$ are independent across $i$, so
\[
\E[\ytrain]{ \Var[\ytest]{ \tfrac{1}{m}\sum_i \big[(\delta_i-\eta_i)^2 - \tfrac{d_i}{1-\epsilon}\big] \agiven \ytrain } }
= \frac{1}{m^2}\sum_{i=1}^m \E[\ytrain]{ \Var[\ytest]{ (\delta_i-\eta_i)^2 \agiven \ytrain } },
\]
since subtracting a constant does not affect variance.

For each $i$,
\begin{align*}
\Var[\ytest]{(\delta_i-\eta_i)^2 \agiven \ytrain}
&= \Var[\ytest]{\delta_i^2 + \eta_i^2 - 2\delta_i\eta_i \agiven \ytrain}\\
&= \Var[\ytest]{\eta_i^2} + 4\delta_i^2 \Var[\ytest]{\eta_i} - 4\delta_i \Cov[\ytest]{\eta_i^2}{\eta_i}.
\end{align*}
Since $\eta_i$ is zero-mean Gaussian, all odd moments vanish, thus
\[
\Cov[\ytest]{\eta_i^2}{\eta_i} = \E[\ytest]{\eta_i^3} - \E[\ytest]{\eta_i^2}\E[\ytest]{\eta_i} = 0.
\]
With $\eta_i \sim \N{0}{\tfrac{d_i}{1-\epsilon}}$ and Lemma~\ref{lem:sqnorm},
\[
\Var[\ytest]{\eta_i^2} = 2\!\left(\tfrac{d_i}{1-\epsilon}\right)^{\!2},
\qquad
\Var[\ytest]{\delta_i\eta_i \agiven \ytrain} = \delta_i^2\,\Var[\ytest]{\eta_i}
= \delta_i^2\,\tfrac{d_i}{1-\epsilon}.
\]
Hence
\[
\Var[\ytest]{(\delta_i-\eta_i)^2 \agiven \ytrain}
= 2\!\left(\tfrac{d_i}{1-\epsilon}\right)^{\!2} + 4\,\delta_i^2\,\tfrac{d_i}{1-\epsilon}.
\]
Taking expectation over $\ytrain$ yields
\[
\E[\ytrain]{ \Var[\ytest]{(\delta_i-\eta_i)^2 \agiven \ytrain} }
= 2\!\left(\tfrac{d_i}{1-\epsilon}\right)^{\!2}
+ 4\,\tfrac{d_i}{1-\epsilon}\,\E[\ytrain]{\delta_i^2}.
\]

\paragraph{Second term.}
\[
\E[\ytest]{(\delta_i-\eta_i)^2 - \tfrac{d_i}{1-\epsilon} \agiven \ytrain}
= \delta_i^2 + \E[\ytest]{\eta_i^2} - \tfrac{d_i}{1-\epsilon}
= \delta_i^2,
\]
where $\E[\ytest]{\eta_i^2}= \tfrac{d_i}{1-\epsilon}$ (Lemma~\ref{lem:sqnorm}). Hence
\[
\Var[\ytrain]{ \E[\ytest]{ \tfrac{1}{m}\sum_{i=1}^m \big[(\delta_i-\eta_i)^2 - \tfrac{d_i}{1-\epsilon}\big] \agiven \ytrain } }
= \Var[\ytrain]{ \tfrac{1}{m}\sum_{i=1}^m \delta_i^2 }.
\]

\paragraph{Result.}
Combining,
\[
\Varthin{\msehat}
= \frac{1}{m^2}\sum_{i=1}^m
\left[
2\!\left(\tfrac{d_i}{1-\epsilon}\right)^{\!2}
+ 4\,\tfrac{d_i}{1-\epsilon}\,\E[\ytrain]{\delta_i^2}
\right]
+ \Var[\ytrain]{ \tfrac{1}{m}\sum_{i=1}^m \delta_i^2 }.
\]
Substituting $\delta_i=(\thetatrain_i - \theta_i)$ and reorganizing yields
\begin{equation} \label{eq:variance}
\Varthin{\msehat}
= \frac{2}{m^2}\sum_{i=1}^m \!\left(
\!\left(\tfrac{d_i}{1-\epsilon}\right)^{\!2}
+ 2\,\tfrac{d_i}{1-\epsilon}\,\E[\ytrain]{(\thetatrain_i - \theta_i)^2}
\right)
+ \Var[\ytrain]{\tfrac{1}{m}\sum_{i=1}^m (\thetatrain_i - \theta_i)^2}.
\end{equation}
 The first term is driven by test-set variability; the second is the variability of squared error across training sets.
\end{proof}


\subsection{Derivation: Variance and Minimizing \texorpdfstring{$\epsilon$}{epsilon} for the Direct Estimator}
\label{app:direct_var_opt}

\begin{lemma}
\label{lem:direct_var_opt}
For the direct estimator $\thetatrain_i = \tfrac{1}{\epsilon}\ytrain_i$, the variance is given by
\[
\Varthin{\msehat} = \frac{2}{m^2}\sum_{i=1}^m \frac{d_i^2}{\epsilon^2(1-\epsilon)^2}
\]
which is minimized at $\epsilon = 1/2$.
\end{lemma}

\begin{proof}
We first show that the variance has the form given above. For the direct estimator, $\thetatrain_i = \frac{1}{\epsilon}\ytrain_i \sim N(\theta_i, d_i/\epsilon)$, so the prediction error satisfies
\[
\thetatrain_i - \theta_i \sim N(0, d_i/\epsilon), \qquad \E[\ytrain]{(\thetatrain_i - \theta_i)^2} = d_i/\epsilon.
\]
By the squared Gaussian lemma and independence across areas, the training variability is
\[
\Var[\ytrain]{\tfrac{1}{m}\sum_{i=1}^m (\thetatrain_i - \theta_i)^2} = \frac{2}{m^2}\sum_{i=1}^m \frac{d_i^2}{\epsilon^2}.
\]
Substituting into the variance formula (Proposition~\ref{prop:mse_variance}):
\begin{align*}
\Varthin{\msehat} &= \frac{2}{m^2}\sum_{i=1}^m \left(\frac{d_i^2}{(1-\epsilon)^2} + \frac{2d_i^2}{\epsilon(1-\epsilon)}\right) + \frac{2}{m^2}\sum_{i=1}^m \frac{d_i^2}{\epsilon^2} \\
&= \frac{2}{m^2}\sum_{i=1}^m d_i^2 \left(\frac{1}{(1-\epsilon)^2} + \frac{2}{\epsilon(1-\epsilon)} + \frac{1}{\epsilon^2}\right).
\end{align*}
The bracketed term simplifies by recognizing it is a square of a sum which further simplifies to
\[
\left( \frac{1}{1-\epsilon} + \frac{1}{\epsilon}  \right)^2 = \left( \frac{1-\epsilon + \epsilon }{\epsilon (1-\epsilon)}\right)^2 = \frac{1}{{\epsilon^2 (1-\epsilon)^2}}.
\]

Since $\sum_{i=1}^m d_i^2$ is constant in $\epsilon$, it suffices to \textit{maximize the denominator} $f(\epsilon) := [\epsilon(1-\epsilon)]^2$ on $(0,1)$. Differentiating,
\[
f'(\epsilon) = 2\epsilon(1-\epsilon)(1 - 2\epsilon) = 0.
\]
The interior critical point is $\epsilon = 1/2$, which is a maximum since $f(\epsilon) \to 0$ as $\epsilon \to 0$ or $\epsilon \to 1$.
\end{proof}


\subsection{
  Proof of Proposition~\ref{prop:fh_var_minimize}:
  Variance-minimizing
  \texorpdfstring{$\epsilon$}{epsilon}
  for Fay--Herriot
}\label{app:fh_var_minimize}
\paragraph{Monotonicity and Minimum for the Fay--Herriot:}
Here we show that under a Fay--Herriot model with known parameters, the variance of the MSE estimator is monotonically increasing for $\epsilon \in [1/2, 1)$ and that the minimum must exist in $(0, 1/2)$.

\begin{proof}
Under the Fay--Herriot model with known $\beta$ and $\sigma^2$, the posterior mean given $\ytrain_i$ is
\[
\tilde{\theta}_i^{(1)} = \gamma_i(\epsilon) \frac{y_i^{(1)}}{\epsilon} + (1-\gamma_i(\epsilon)) x_i^\top \beta
\]
where $\gamma_i(\epsilon) = \epsilon\sigma^2 / (\epsilon\sigma^2 + d_i)$.

Recall that the model assumes $\theta_i = x_i^\top\beta + u_i$ where $u_i$ are the IID random effects. Using this we derive the prediction error for area $i$:
\begin{align*}
\tilde{\theta}_i^{(1)} - \theta_i
&= \gamma_i(\epsilon) \frac{y_i^{(1)}}{\epsilon} + (1-\gamma_i(\epsilon)) x_i^\top\beta - (x_i^\top\beta + u_i) \\
&= \gamma_i(\epsilon) \left(\frac{y_i^{(1)}}{\epsilon} - x_i^\top \beta\right) - u_i \\
&= \gamma_i(\epsilon) \left(\frac{y_i^{(1)}}{\epsilon} - \theta_i + u_i \right) - u_i \\
&= \gamma_i(\epsilon) \left(\frac{y_i^{(1)}}{\epsilon} - \theta_i \right) -(1-\gamma_i(\epsilon)) u_i ,
\end{align*}
Note that $u_i \iid N(0, \sigma^2)$ and $y_i^{(1)}/\epsilon \ind N(\theta_i, d_i/\epsilon)$. Thus the prediction error is a weighted combination of two independent zero-mean Gaussian distributions with the combined variance
\begin{align*}
    g_i(\epsilon) \, &:=  \gamma_i(\epsilon)^2 \frac{d_i}{\epsilon}  + (1-\gamma_i(\epsilon))^2 \sigma^2 \\
    &= \left(\frac{\epsilon\sigma^2}{\epsilon\sigma^2 + d_i}\right)^2 \frac{d_i}{\epsilon} + \left(\frac{d_i}{\epsilon\sigma^2 + d_i}\right)^2 \sigma^2 = \sigma^2d_i \cdot \frac{\epsilon\sigma^2 + d_i}{(\epsilon\sigma^2 + d_i)^2} = \frac{\sigma^2 d_i}{(\epsilon\sigma^2 + d_i)}.
\end{align*}

Moreover, the prediction errors are independent across areas. By Lemma~\ref{lem:sqnorm}, the training variability is
\[
\Var[\ytrain]{\tfrac{1}{m}\sum_{i=1}^m (\tilde{\theta}_i^{(1)} - \theta_i)^2} = \frac{1}{m^2}\sum_{i=1}^m \Var{(\tilde{\theta}_i^{(1)} - \theta_i)^2} = \frac{2}{m^2}\sum_{i=1}^m g_i(\epsilon)^2.
\]
Substituting into the variance formula from Proposition~\ref{prop:mse_variance}, the variance becomes
\[
V(\epsilon) = \frac{2}{m^2}\sum_{i=1}^m \left[\frac{d_i^2}{(1-\epsilon)^2} + 2\frac{d_i}{1-\epsilon} \cdot g_i(\epsilon) + g_i(\epsilon)^2\right] = \frac{2}{m^2}\sum_{i=1}^m \left[\frac{d_i}{1-\epsilon} + g_i(\epsilon)\right]^2.
\]
Define $f_i(\epsilon) := \frac{d_i}{1-\epsilon} + g_i(\epsilon)$. Then $V(\epsilon) = \frac{2}{m^2}\sum_{i=1}^m f_i(\epsilon)^2$ and
\[
V'(\epsilon) = \frac{4}{m^2}\sum_{i=1}^m f_i(\epsilon) \cdot f_i'(\epsilon).
\]
Since $f_i(\epsilon) > 0$ for all $\epsilon \in (0,1)$, it suffices to show $f_i'(\epsilon) > 0$ for $\epsilon \in [1/2, 1)$.

Differentiating,
\[
f_i'(\epsilon) = \frac{d_i}{(1-\epsilon)^2} + g_i'(\epsilon) = \frac{d_i}{(1-\epsilon)^2} - \frac{\sigma^4 d_i}{(\epsilon\sigma^2 + d_i)^2}.
\]
Combining over a common denominator,
\[
f_i'(\epsilon) = \frac{d_i \left[(\epsilon\sigma^2 + d_i)^2 - \sigma^4(1-\epsilon)^2\right]}{(1-\epsilon)^2(\epsilon\sigma^2 + d_i)^2}.
\]
The numerator inside the brackets factors as
\begin{align*}
(\epsilon\sigma^2 + d_i)^2 - \sigma^4(1-\epsilon)^2
&= d_i^2 - \sigma^4 + 2\epsilon\sigma^2(d_i + \sigma^2) \\
&= (d_i + \sigma^2)(d_i - \sigma^2 + 2\epsilon\sigma^2) \\
&= (d_i + \sigma^2)(d_i + (2\epsilon - 1)\sigma^2).
\end{align*}
Thus
\[
f_i'(\epsilon) = \frac{d_i(d_i + \sigma^2)(d_i + (2\epsilon - 1)\sigma^2)}{(1-\epsilon)^2(\epsilon\sigma^2 + d_i)^2}.
\]

For $\epsilon \geq 1/2$, we have $(2\epsilon - 1) \geq 0$, so $d_i + (2\epsilon - 1)\sigma^2 > 0$. Note that $d_i, \sigma^2 >0$ and the denominator is strictly positive as well. Hence $f_i'(\epsilon) > 0$ for all $i$ and all $\epsilon \in [1/2, 1)$.

Therefore $V'(\epsilon) > 0$ on $[1/2, 1)$, establishing that $V(\epsilon)$ is strictly increasing on this interval.

Also as $\epsilon$ approaches $1$, the test-set term $d_i/(1-\epsilon) \to \infty$, so $V(\epsilon) \to \infty$. On the other side, as $\epsilon \to 0^+$, both $d_i/(1-\epsilon) \to d_i$ and $g_i(\epsilon) \to \sigma^2$ remain bounded, so $V(\epsilon)$ is bounded.

Since $V$ is continuous on $(0,1)$, the variance-minimizing $\epsilon^*$ lies strictly below $1/2$.
\end{proof}

\paragraph{Variance-minimizing $\epsilon_i^*$ for each area:}

Now we simply use the derivative above to find a variance-minimizing $\epsilon_i^*$ for each area $i=1, \ldots, m.$
\begin{proof}
The area-specific variance contribution is proportional to $f_i(\epsilon)^2$. Since $f_i(\epsilon) > 0$, minimizing $f_i(\epsilon)^2$ is equivalent to finding where $f_i'(\epsilon) = 0$. From the factored form
\[
f_i'(\epsilon) = \frac{d_i(d_i + \sigma^2)(d_i + (2\epsilon - 1)\sigma^2)}{(1-\epsilon)^2(\epsilon\sigma^2 + d_i)^2},
\]
the only root in $(0,1)$ occurs when $d_i + (2\epsilon - 1)\sigma^2 = 0$, yielding
\[
\epsilon_i^* = \frac{1}{2} - \frac{d_i}{2\sigma^2}.
\]
This is a minimum since $f_i'(\epsilon) < 0$ for $\epsilon < \epsilon_i^*$ and $f_i'(\epsilon) > 0$ for $\epsilon > \epsilon_i^*$.

Given the range $\epsilon$, the variance-minimizing value truncated to the feasible range is
\[
\epsilon_i^* \;=\; \max\!\left\{\,0,\; \frac{1}{2} - \frac{d_i}{2\sigma^2}\right\}.
\]
\end{proof}

\subsection{Log-Scale Version of Figure~\ref{fig:tradeoff}}\label{app:log_tradeoff}

\begin{figure}[H]
    \centering
  \includegraphics[width=0.9\linewidth]{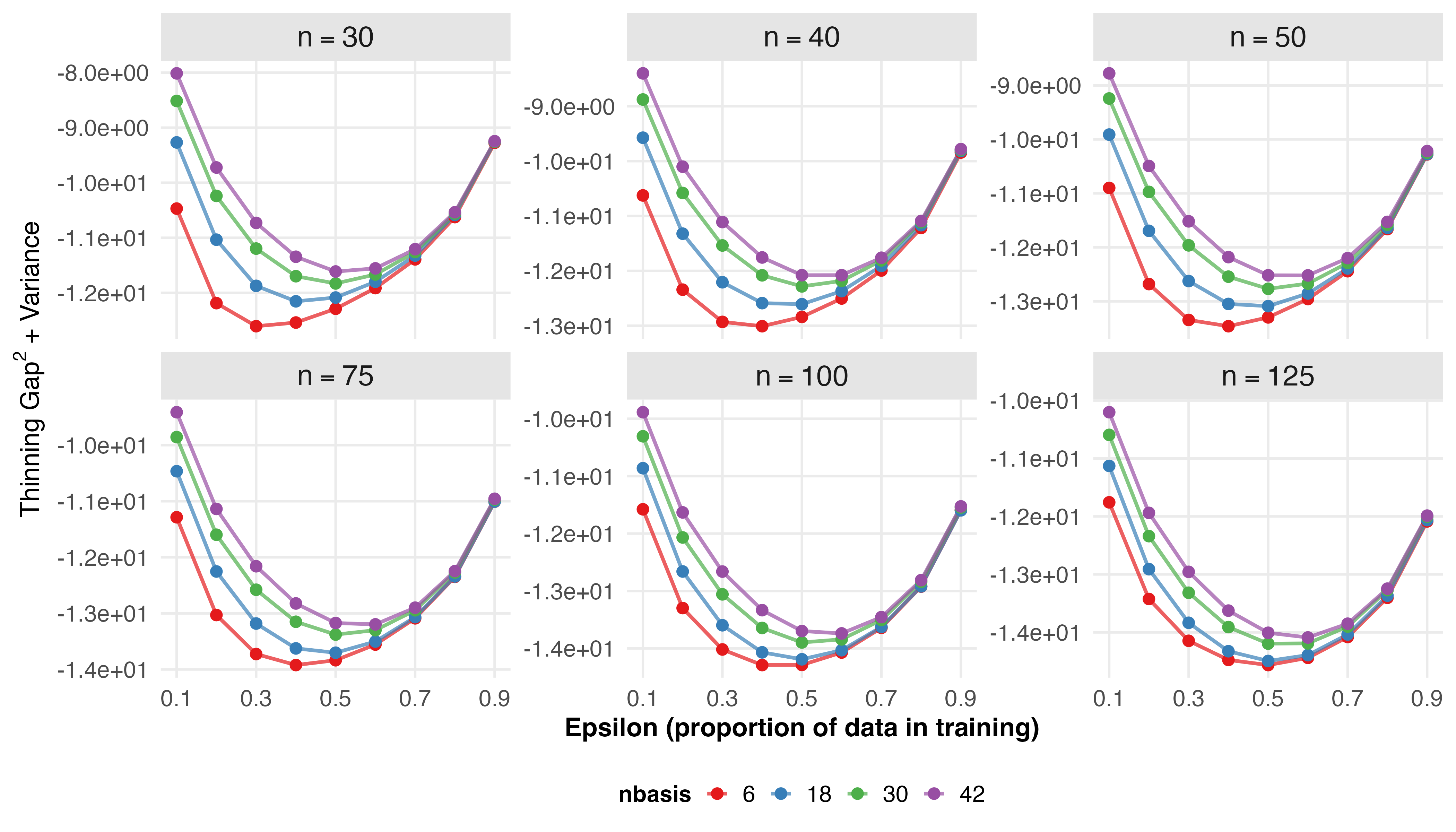}
    \caption{
    The thinning gap-variance trade-off: sum of squared thinning gap and variance of the MSE estimator for Fay--Herriot models with $p = 6, 18, 30, 42$ spatial basis functions averaged across 50 samples from each design. The log-scale reveals the differing interior optima for each model and how the gap in the curve shrinks with higher $\epsilon$.}
  \end{figure}

\subsection{Extrapolation-Based Thinning Gap Correction}\label{app:extrapolation}

Recall that the thinning gap is the systematic difference $\msethin - \msefull = m^{-1}\sum_i \Delta_i(\epsilon)$ between the thinned-data oracle MSE and the full-data target (Proposition~\ref{prop:thinning_gap}). We reparametrize the per-area gap $\Delta_i(\epsilon)$ in terms of $t = (1-\epsilon)/\epsilon$, so that $t = 0$ corresponds to the full-data limit $\epsilon = 1$. Writing $\kappa_i = 1 - \gamma_i = d_i/(\sigma^2 + d_i)$ and substituting $\epsilon = 1/(1+t)$,
\[
\Delta_i(\epsilon) = \gamma_i^2\, d_i \cdot \frac{t}{1 + \kappa_i\, t},
\]
a rational function of $t$ that vanishes at $t = 0$ and is increasing and concave. Expanding $1/(1 + \kappa_i t) = \sum_{k=0}^\infty (-\kappa_i t)^k$ for $|\kappa_i t| < 1$, the average gap admits the power series
\[
\msethin - \msefull = \frac{1}{m}\sum_{i=1}^m \Delta_i(\epsilon) = c_1\, t + c_2\, t^2 + c_3\, t^3 + \cdots,
\]
with $c_1 = m^{-1}\sum_i \gamma_i^2 d_i > 0$ and higher-order coefficients that alternate in sign with decreasing magnitude since $\kappa_i \in (0,1)$. The linear term therefore dominates, and Figure~\ref{fig:extrapolation_gap} confirms that the gap is approximately linear in $t$ across all designs. The parameter-uncertainty correction $g_{2i}(\epsilon)$ from Proposition~\ref{prop:thinning_gap2}, which accounts for estimation of $\beta$, is $O(p/m)$ and likewise rational in $t$, so this structure is preserved under estimated parameters.

\begin{figure}[H]
    \centering
  \includegraphics[width=0.9\linewidth]{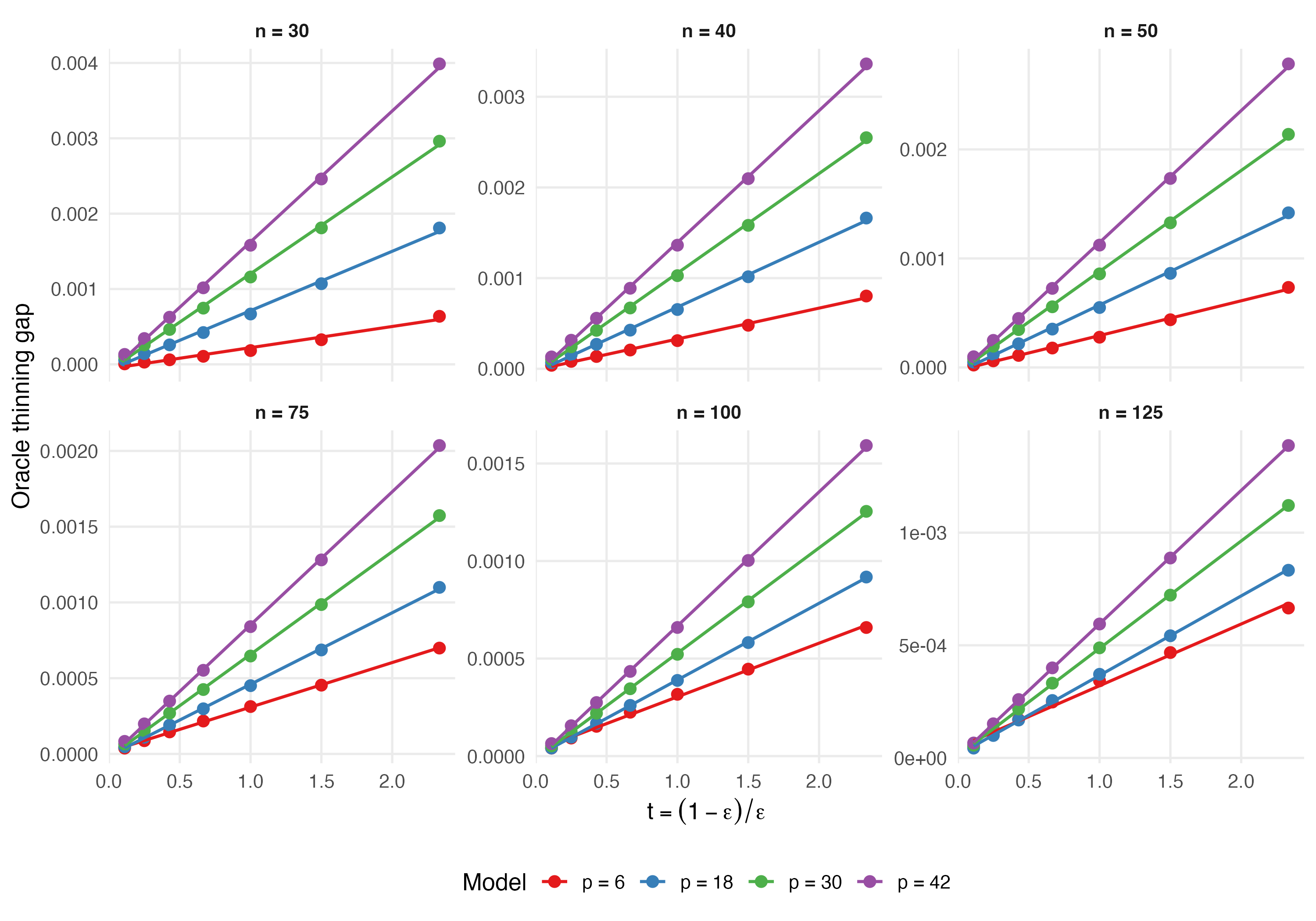}
    \caption{The thinning gap $\msethin - \msefull$, computed at the true $\theta_i$ and averaged over $S = 50$ samples, plotted against $t = (1-\epsilon)/\epsilon$ over the extrapolation-grid range ($\epsilon \geq 0.3$) for $p = 6, 18, 30, 42$ basis functions. The approximately linear relationship in $t$ is consistent across all designs.}
    \label{fig:extrapolation_gap}
\end{figure}

Because $\msehat$ is unbiased for $\msethin$ (Theorem~\ref{thm:unbiased_mse}) and $\msethin = \msefull$ at $t = 0$, the near-linearity suggests estimating $\msefull$ by evaluating $\msehat$ on a grid $\epsilon_1, \ldots, \epsilon_K$, fitting the polynomial regression
\[
\msehat(\epsilon_k) = \alpha_0 + \alpha_1\, t_k + \alpha_2\, t_k^2 + \eta_k,
\]
and reading off the intercept $\hat\alpha_0$ as an extrapolated estimate of $\msefull$ at $t = 0$ ($\epsilon = 1$). 
With the feasible, noisy $\msehat$, however, the intercept is anchored by the grid points nearest $\epsilon = 1$, where the test set is smallest and $\Var{\msehat}$ is largest (\S\ref{sec:variance}); the resulting per-area test-proxy variance $d_i/(1-\epsilon)$ overwhelms the extrapolated signal. Table~\ref{tab:extrapolation} compares pooled estimators using $R = 5$ thinning replicates (quadratic and linear fits over two grids) against single-$\epsilon$ baselines, scored by the RMSE of the selected basis count relative to the average oracle $p^*$.

\begin{table}[H]
\centering
\caption{RMSE of selected $p$ versus average oracle $p^*$ ($R = 5$, DT-MSE). Pooled extrapolation methods (top) versus single-$\epsilon$ baselines (bottom).}
\label{tab:extrapolation}
\begin{tabular}{l rrrrrr}
\toprule
Method & $n\!=\!30$ & $n\!=\!40$ & $n\!=\!50$ & $n\!=\!75$ & $n\!=\!100$ & $n\!=\!125$ \\
\midrule
Quadratic $[0.3, 0.7]$ & 8.2 & 8.3 & 11.5 & 13.6 & 14.3 & 14.3 \\
Linear $[0.3, 0.7]$    & 8.3 & 9.1 & 10.8 & 12.6 & 15.2 & 16.1 \\
Linear $[0.5, 0.7]$    & 8.2 & 8.3 & 11.0 & 13.6 & 14.2 & 14.3 \\
\midrule
$\epsilon = 0.5$ & 7.9 & 9.9 & 8.8 & 7.8 &  6.1 &  4.3 \\
$\epsilon = 0.6$ & 7.5 & 9.6 & 8.3 & 7.5 &  5.6 &  6.3 \\
$\epsilon = 0.7$ & 7.6 & 9.0 & 8.0 & 8.2 &  9.5 &  9.6 \\
\bottomrule
\end{tabular}
\end{table}

The pooled estimator underperforms the best single-$\epsilon$ baseline in 5 of the 6 designs. The degradation is most pronounced at the large sample sizes ($n = 100, 125$), where the spread of the thinned MSE across candidate models is smallest and the extrapolation noise is therefore most damaging relative to the signal. The narrower $[0.5, 0.7]$ grid with a linear fit performs no worse than the quadratic fit over $[0.3, 0.7]$, so neither adding low-$\epsilon$ points nor using a quadratic fit brings any benefit. Extending the grid in the other direction does not help either: including $\epsilon = 0.8$, which lies closer to the target $t = 0$, degrades the fit across all designs, as its larger variance outweighs the shorter extrapolation distance.


\subsection{Multi-fold Gaussian Data Thinning}\label{app:multifold}

Multi-fold thinning generalizes Algorithm~\ref{alg:thin} to produce $K \geq 2$ mutually independent folds of each direct estimate. The marginal distribution of each fold is $y_i^{(k)} \sim \N{\theta_i/K}{d_i/K}$, the folds are mutually independent across $k$, and they sum to $y_i$. These properties hold only marginally, not conditionally on $y_i$.

\begin{algorithm}[H]
\caption{Multi-fold Gaussian Data Thinning (equal folds; based on Algorithm~2 and Example~5 of \citealt{neufeld_data_2024})}\label{alg:multifold}
\begin{algorithmic}[1]
\Require Direct estimates $y_i \sim \N{\theta_i}{d_i}$ with known variances $d_i$, for $i = 1, \ldots, m$
\Require Number of folds $K \geq 2$
\For{each area $i = 1, \ldots, m$}
  \State Draw $(y_i^{(1)}, \ldots, y_i^{(K)}) \mid y_i \;\sim\; N_{K}\!\left(\frac{1}{K}\, y_i\,\mathbf{1}_{K},\; \frac{1}{K}\, d_i\!\left(I_{K} - \frac{1}{K}\,\mathbf{1}_{K}\mathbf{1}_{K}^\top\right)\right)$
  \Statex \hspace{2.5em} subject to $\textstyle\sum_{k=1}^K y_i^{(k)} = y_i$
  \State Set $y_i^{(-k)} \leftarrow y_i - y_i^{(k)}$ for each $k = 1, \ldots, K$
\EndFor
\State \Return $\bigl\{y_i^{(k)},\, y_i^{(-k)}\bigr\}_{k=1}^{K}$ for each area $i$
\end{algorithmic}
\end{algorithm}

For equal folds, the joint conditional distribution of $(y_i^{(1)}, \ldots, y_i^{(K)}) \mid y_i$ is a degenerate multivariate normal (Example~5 of \citealt{neufeld_data_2024}) and the constraint $\textstyle\sum_{k=1}^K y_i^{(k)} = y_i$ is needed.

For each fold $k = 1, \ldots, K$, define the training component $y_i^{(-k)} := y_i - y_i^{(k)}$, which has marginal distribution $y_i^{(-k)} \sim \N{(1-\epsilon)\,\theta_i}{(1-\epsilon)\,d_i}$ with $\epsilon = (K-1)/K$. Note that one could allocate more than one fold for testing to adjust training fraction given a fixed $K$. Ex: for $K=5$, use $3$ components for training and 2 components for testing to set training fraction $\epsilon=3/5$ for each fold.


\subsection{Proof of Weighted MSE Equivalence of Likelihood Validation}
\label{app:lik_mse}

\begin{proof}
The plug-in predictive log-likelihood is
\[
\ell_{\epsilon} = \sum_{i=1}^m \log \phi\!\left( \ytest_i \agiven (1-\epsilon)\thetatrain_i,\, (1-\epsilon)d_i \right).
\]
Expanding the Gaussian log-density,
\[
\ell_{\epsilon} = \sum_{i=1}^m \left[ -\tfrac{1}{2}\log(2\pi(1-\epsilon)d_i) - \frac{(\ytest_i - (1-\epsilon)\thetatrain_i)^2}{2(1-\epsilon)d_i} \right].
\]
Conditioning on the training set $\ytrain$ and taking expectations over $\ytest$, we evaluate the squared term. Write $\ytest_i = (1-\epsilon)\theta_i + (1-\epsilon)\eta_i$ where $\eta_i \sim \N{0}{d_i/(1-\epsilon)}$. Then
\begin{align*}
\ytest_i - (1-\epsilon)\thetatrain_i
&= (1-\epsilon)\theta_i + (1-\epsilon)\eta_i - (1-\epsilon)\thetatrain_i \\
&= (1-\epsilon)\left( \theta_i - \thetatrain_i + \eta_i \right).
\end{align*}
Thus
\[
\left( \ytest_i - (1-\epsilon)\thetatrain_i \right)^2 = (1-\epsilon)^2 \left( \theta_i - \thetatrain_i + \eta_i \right)^2.
\]
Taking expectations over $\ytest$ (i.e., over $\eta_i$) with $\ytrain$ fixed,
\begin{align*}
\E[\ytest]{\left( \ytest_i - (1-\epsilon)\thetatrain_i \right)^2 \agiven \ytrain}
&= (1-\epsilon)^2 \E[\ytest]{\left( \theta_i - \thetatrain_i + \eta_i \right)^2 \agiven \ytrain} \\
&= (1-\epsilon)^2 \left[ \left( \thetatrain_i - \theta_i \right)^2 + \Var{\eta_i} \right] \\
&= (1-\epsilon)^2 \left[ \left( \thetatrain_i - \theta_i \right)^2 + \tfrac{d_i}{1-\epsilon} \right],
\end{align*}
where the cross-term vanishes since $\E{\eta_i} = 0$. Substituting back,
\begin{align*}
\E[\ytest]{\ell_{\epsilon} \given \ytrain}
&= \sum_{i=1}^m \left[ -\tfrac{1}{2}\log(2\pi(1-\epsilon)d_i) - \frac{(1-\epsilon)^2}{2(1-\epsilon)d_i} \left( \left( \thetatrain_i - \theta_i \right)^2 + \tfrac{d_i}{1-\epsilon} \right) \right] \\
&= \sum_{i=1}^m \left[ -\tfrac{1}{2}\log(2\pi(1-\epsilon)d_i) - \tfrac{1-\epsilon}{2d_i} \left( \thetatrain_i - \theta_i \right)^2 - \tfrac{1}{2} \right] \\
&= C - \frac{1}{2}\sum_{i=1}^m \frac{1-\epsilon}{d_i} \left( \thetatrain_i - \theta_i \right)^2,
\end{align*}
where $C = -\tfrac{m}{2} - \tfrac{1}{2}\sum_{i=1}^m \log(2\pi(1-\epsilon)d_i)$ depends only on known constants.
\end{proof}

\end{document}